\documentclass{myllncs}
\usepackage{graphicx}
\usepackage{amsmath}
\usepackage{amssymb}
\usepackage[T1]{fontenc}
\usepackage{enumerate}
\usepackage[font=small, format=hang, labelfont=bf]{caption}
\usepackage[]{caption}
\usepackage{subfig}

\newcommand{\NP}{\mathcal{NP}}
\newcommand{\eat}[1] {{}}
\newcommand{\quaa}{$\mathcal{Q}_1$}
\newcommand{\quab}{$\mathcal{Q}_2$}
\newcommand{\quai}{$\mathcal{Q}_i$}
\newcommand{\quaaa}{$\mathcal{Q}_1^1$}
\newcommand{\quaab}{$\mathcal{Q}_1^2$}

\newcommand{\qua}{$\mathcal{Q}$}

\begin{document}

\title{The Straight-Line RAC Drawing Problem is NP-Hard}

\author{Evmorfia N.\ Argyriou \inst{1}, Michael A.\ Bekos \inst{1}, and Antonios
Symvonis \inst{1}}

\authorrunning{E.N. Argyriou, M.A. Bekos, A. Symvonis}

\tocauthor{Evmorfia N. Argyriou, Michael A. Bekos, Antonios Symvonis}

\authorrunning{E.N. Argyriou, M.A. Bekos, A. Symvonis}

\institute{%
    School of Applied Mathematical \& Physical Sciences,\\
    National Technical University of Athens, Greece\\
    \email{$\{$fargyriou,mikebekos,symvonis$\}$@math.ntua.gr}
}

\maketitle              

\begin{abstract}
Recent cognitive experiments have shown that the negative impact of
an edge crossing on the human understanding of a graph drawing,
tends to be eliminated in the case where the crossing angles are
greater than $70$ degrees. This motivated the study of \emph{RAC
drawings}, in which every pair of crossing edges intersects at right
angle. In this work, we demonstrate a class of graphs with unique
RAC combinatorial embedding and we employ members of this class in
order to show that it is $\NP$-hard to decide whether a graph admits
a straight-line RAC drawing.\\~\\Date: \emph{September 27, 2010.}
\end{abstract}

\section{Introduction}
\label{sec:introduction}

In the graph drawing literature, the problem of finding
aesthetically pleasant drawings of graphs has been extensively
studied. The graph drawing community has introduced and studied
several criteria that judge the quality of a graph drawing, such as
the number of crossings among pairs of edges, the number of edge
bends, the maximum edge length, the total area occupied by the
drawing and so on (see the books~\cite{DBTT94,KW01}).

Motivated by the fact that the edge crossings have negative impact
on the human understanding of a graph drawing
\cite{P00,PCA02,CPCM02}, a great amount of research effort has been
devoted on the problem of finding drawings with minimum number of
edge crossings. Unfortunately, this problem is $\NP$-complete in
general \cite{GJ83}. However, recent eye-tracking experiments by
Huang et al.\ \cite{Hu07,HHE08} indicate that the negative impact of
an edge crossing is eliminated in the case where the crossing angle
is greater than $70$ degrees. These results motivated the study of a
new class of drawings, called \emph{right-angle drawings} or
\emph{RAC drawings} for short \cite{ACBDFKS09,DGDLM10,DEL09,DEL10}.
A RAC drawing of a graph is a polyline drawing in which every pair
of crossing edges intersects at right angle.

Didimo, Eades and Liota \cite{DEL09} proved that it is always
feasible to construct a RAC drawing of a given graph with at most
three bends per edge. In this work, we prove that the problem of
determining whether an input graph admits a straight-line RAC
drawing is $\NP$-hard.

\subsection{Related Work}
\label{sec:related-work}

Didimo et al.\ \cite{DEL09} initiated the study of RAC drawings and
showed that any straight-line RAC drawing with $n$ vertices has at
most $4n-10$ edges and that any graph admits a RAC drawing with at
most three bends per edge. A slightly weaker bound on the number of
edges of an $n$-vertices RAC drawing was given by Arikushi et al.
\cite{AFKMT10}, who proved that any straight-line RAC drawing with
$n$ vertices may have $4n-8$ edges. Angelini et al.\
\cite{ACBDFKS09} showed that the problem of determining whether an
acyclic planar digraph admits a straight-line upward RAC drawing is
$\NP$-hard. Furthermore, they constructed digraphs admitting
straight-line upward RAC drawings, that require exponential area. Di
Giacomo et al.\ \cite{DGDLM10} studied the interplay between the
crossing resolution, the maximum number of bends per edges and the
required area. Didimo et al.\ \cite{DEL10} presented a
characterization of complete bipartite graphs that admit a
straight-line RAC drawing. Arikushi et al.\ \cite{AFKMT10} studied
polyline RAC drawings in which each edge has at most one or two
bends and proved that the number of edges is at most $O(n)$ and
$O(n\log^2{n})$, respectively. Dujmovic et al.\ \cite{DGMW10}
studied \emph{$\alpha$~Angle Crossing} (or \emph{$\alpha$AC} for
short) drawings, i.e., drawings in which the smallest angle formed
by an edge crossing is at least $\alpha$. In their work, they
presented upper and lower bounds on the number of edges. Van Kreveld
\cite{vK10} studied how much better (in terms of area required,
edge-length and angular resolution) a RAC drawing of a planar graph
can be than any planar drawing of the same graph.

Closely related to the RAC drawing problem, is the angular
resolution maximization problem, i.e., the problem of maximizing the
smallest angle formed by any two adjacent edges incident to a common
vertex. Note that both problems correlate the resolution of a graph
with the visual distinctiveness of the edges in a graph drawing.
Formann et al.\ \cite{FHHKLSWW93} introduced the notion of the
angular resolution of straight-line drawings. In their work, they
proved that determining whether a graph of maximum degree $d$ admits
a drawing of angular resolution $\frac{2 \pi}{d}$ (i.e., the obvious
upper bound) is $\NP$-hard. They also presented upper and lower
bounds on the angular resolution for several types of graphs of
maximum degree $d$. Malitz and Papakostas \cite{MP92} proved that
for any planar graph of maximum degree $d$, it is possible to
construct a planar straight-line drawing with angular resolution
$\Omega(\frac{1}{7^d})$.
Garg and Tamassia \cite{GT94} presented a continuous tradeoff
between the area and the angular resolution of planar straight-line
drawings. For the case of connected planar graphs with $n$ vertices
and maximum degree $d$, Gutwenger and Mutzel \cite{GM98} presented a
linear time algorithm that constructs planar polyline grid drawings
on a $(2n-5)\times(\frac{3}{2}n-\frac{7}{2})$ grid with at most
$5n-15$ bends and minimum angle greater than $\frac{2}{d}$.
Bodlaender and Tel \cite{BT04} showed that planar graphs with
angular resolution at least $\frac{\pi}{2}$ are rectilinear.
Recently, Lin and Yen \cite{LY05} presented a force-directed
algorithm based on edge-edge repulsion that constructs drawings with
high angular resolution. Argyriou et al.\ \cite{ABS10} studied a
generalization of the crossing and angular resolution maximization
problems, in which the minimum of these quantities is maximized and
presented optimal algorithms for complete graphs and a
force-directed algorithm for general graphs.

The rest of this paper is structured as follows: In
Section~\ref{sec:preliminaries}, we introduce preliminary properties
and notation. In Section~\ref{sec:racgraphs}, we present a class of
graphs with unique RAC combinatorial embedding. In
Section~\ref{sec:np}, we show that the straight-line RAC drawing
problem is $\NP$-hard. We conclude in Section~\ref{sec:conclusions}
with open problems.

\section{Preliminaries}
\label{sec:preliminaries}

Let $G=(V,E)$ be a simple, undirected graph drawn in the plane. We
denote by $\Gamma(G)$ the drawing of $G$. Given a drawing
$\Gamma(G)$ of a graph $G$, we denote by $\ell_{u,v}$ the line
passing through vertices $u$ and $v$. By $\ell_{u,v}'$, we refer to
the semi-line that emanates from vertex $u$, towards vertex $v$.
Similarly, we denote by $\ell_{u,v,w}$ ($\ell_{u,v,w}'$) the line
(semi-line) that coincides (emanates from) vertex $u$ and is
perpendicular to edge $(v,w)$. The following properties are used in
the rest of this paper.

\begin{property}[Didimo, Eades and Liota \cite{DEL09}]
In a straight-line RAC drawing there cannot be three mutually
crossing edges. \label{prp:three-crossing-edges}
\end{property}

\begin{property}[Didimo, Eades and Liota \cite{DEL09}]
In a straight-line RAC drawing there cannot be a triangle
$\mathcal{T}$ and two edges $(a,b)$ and $(a,b')$, such that $a$ lies
outside $\mathcal{T}$ and $b$, $b'$ lie inside $\mathcal{T}$.
\label{prp:triangle-edges}
\end{property}

\section{A Class of Graphs with Unique RAC Combinatorial Embedding}
\label{sec:racgraphs}

The $\NP$-hardness proof employs a reduction from the well-known
$3$-SAT problem \cite{GJ79}. However, before we proceed with the
reduction details, we first provide a graph, referred to as
\emph{augmented square antiprism graph}, which has the following
property: All RAC drawings of this graph have two ``symmetric''
combinatorial embeddings. Figures~\ref{fig:basic-gadget-2} and
\ref{fig:basic-gadget-2_2} illustrate this property. Observe that the augmented
square antiprism graph consists of a ``central'' vertex $v_0$, which
is incident to all vertices of the graph, and two quadrilaterals
(refer to the dashed and bold drawn squares in
Figure~\ref{fig:basic-gadget-2_2}), that are denoted by \quaa~and
\quab~in the remainder of this paper. Removing the central vertex,
the remaining graph corresponds to the skeleton of a square
antiprism, and, it is commonly referred to as \emph{square antiprism
graph}.

\begin{figure}[h!tb]
  \centering
  \begin{minipage}[b]{.32\textwidth}
     \centering
     \subfloat[\label{fig:basic-gadget-2}{}]
     {\includegraphics[width=\textwidth]{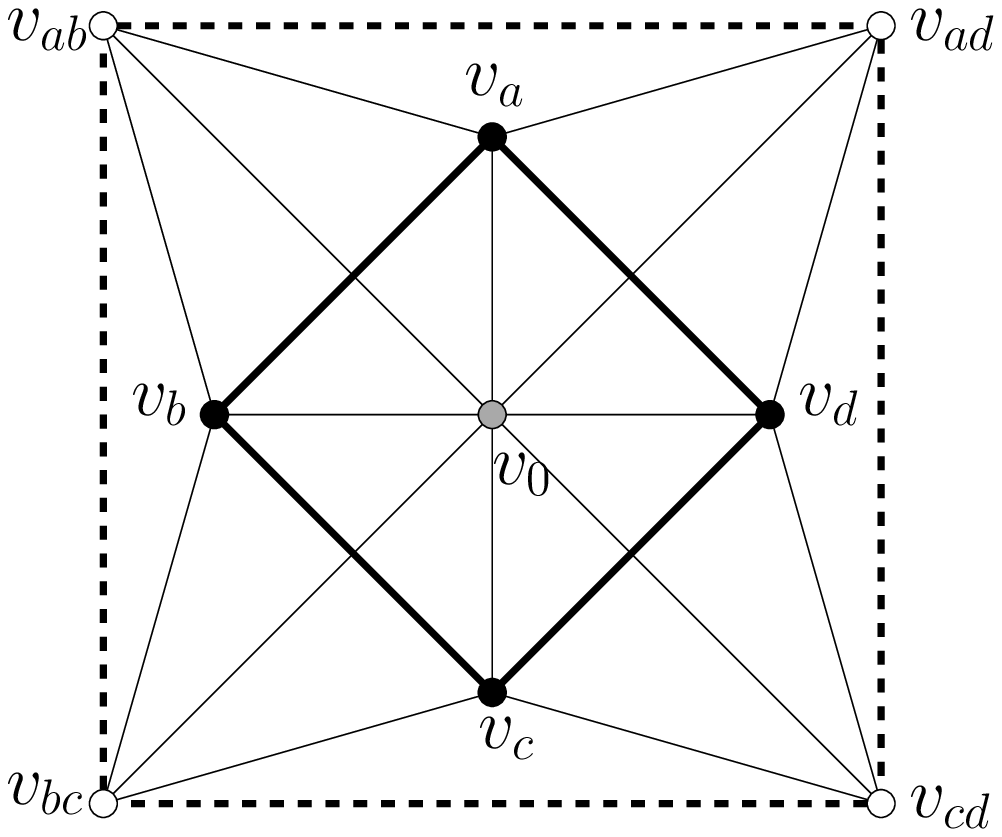}}
  \end{minipage}
  \hfill
  \begin{minipage}[b]{.32\textwidth}
     \centering
     \subfloat[\label{fig:basic-gadget-2_2}{}]
     {\includegraphics[width=\textwidth]{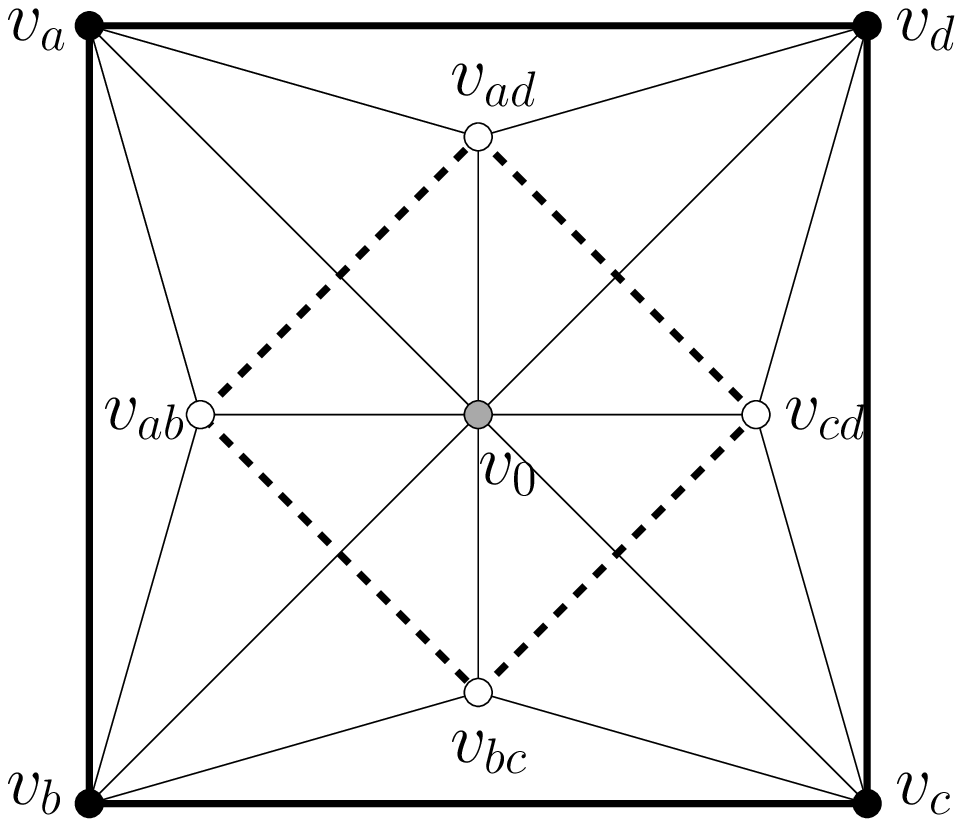}}
  \end{minipage}
  \begin{minipage}[b]{.32\textwidth}
     \centering
     \subfloat[\label{fig:basic-gadget-3}{}]
     {\includegraphics[width=\textwidth]{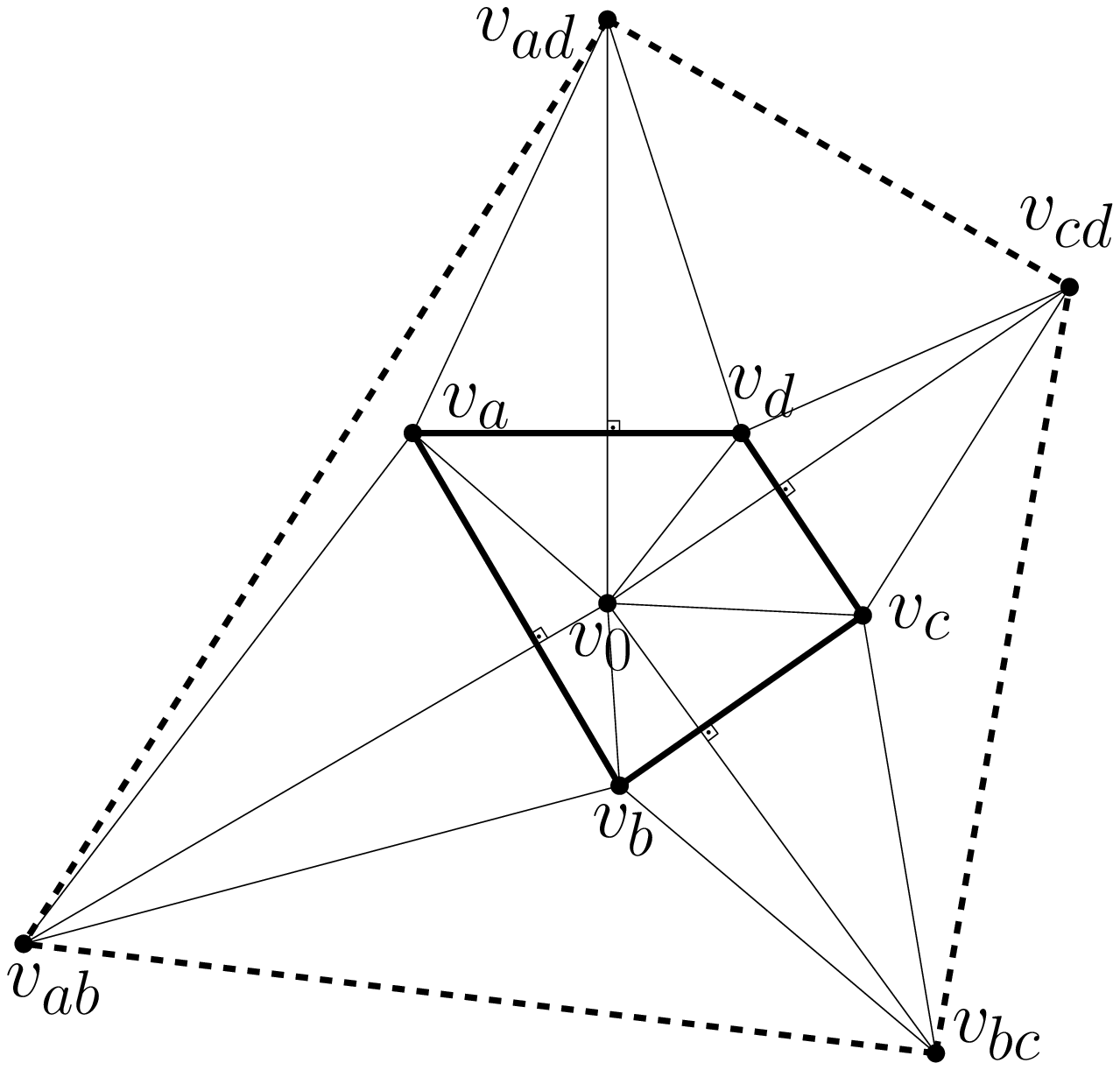}}
  \end{minipage}
  \caption{(a)-(b)~Two different RAC drawings of the augmented square antiprism graph with different combinatorial embeddings.
  (a)-(c) Two different RAC drawing with the same combinatorial embedding.}
  \label{fig:basic-gadget-1}
\end{figure}

If we replace the two quadrilaterals with two triangles, then the
implied graph is the \emph{augmented triangular antiprism graph}.
Didimo et al.\ \cite{DEL09}, who showed that any $n$-vertex graph
which admits a RAC-drawing can have at most $4n-10$ edges, used the
augmented triangular antiprism graph, as an example of a graph that
achieves the bound of $4n-10$ edges (see Figure~1.c in
\cite{DEL09}). In contrast to the augmented triangular antiprism
graph, the augmented square antiprism graph does not achieve this
upper bound. In general, the class of \emph{the augmented $k$-gon
antiprism graphs, $k \geq 3$}, is a class of non-planar graphs, that
all admit RAC drawings. Recall that any planar $n$-vertices graph,
should have $3n-6$ edges, and since an augmented $k$-gon antiprism
graph has $2k+1$ vertices and $5k$ edges, it is not planar for the
entire class of these graphs.


\begin{lemma}
There does not exist a RAC drawing of the augmented square antiprism
graph in which the central vertex $v_0$ lies on the exterior of
quadrilateral \quai, $i=1,2$, and an edge connecting $v_0$ with a
vertex of \quai~crosses an edge of \quai. \label{lem:v0nocross}
\end{lemma}

\begin{proof}
Let \qua~be one of quadrilaterals \quai, $i=1,2$ and let $v_a$,
$v_b$, $v_c$ and $v_d$ be its vertices, consecutive along
quadrilateral \qua. Assume to the contrary that vertex $v_0$ lies on
the exterior of quadrilateral \qua~and there exists an edge, say
$(v_0,v_a)$, that emanates from vertex $v_0$ towards a vertex of
quadrilateral \qua, such that it crosses an edge, say
$(v_b,v_c)$\footnote{The case where, it crosses edge $(v_c,v_d)$ is
symmetric.}, of quadrilateral \qua. Vertices $v_b$ and $v_c$ have
the following properties: (a)~they are both connected to vertex
$v_0$, and, (b)~have a common neighbor $v_{bc}$, which is incident
to vertex $v_0$ and $v_{bc} \notin\ \mathcal{Q}$~(see
Figure~\ref{fig:basic-gadget-1}).

Observe that if vertex $v_{bc}$ lies in the non-colored regions of
Figure~\ref{fig:v0nocross-regions}, then at least one of the edges
incident to $v_{bc}$ crosses either $(v_0,v_a)$ or $(v_b,v_c)$,
which are already involved in a right-angle crossing. This leads to
a situation where three edges mutually cross, which, by
Property~\ref{prp:three-crossing-edges} is not permitted. Hence,
vertex $v_{bc}$ should lie in the interior of the dark-gray colored
regions $R_1$, $R_2$ or $R_3$ of Figure~\ref{fig:v0nocross-regions}.
We consider each of these cases separately in the following. Note
that, there exist cases where $R_2 \cup R_3= \emptyset$ (i.e.,
vertex $v_0$ is close to the intersection point of $(v_0,v_a)$ and
$(v_b,v_c)$), or $R_2 = \emptyset$ (i.e., vertex $v_c$ is close to
the intersection point of $(v_0,v_a)$ and $(v_b,v_c)$), or $R_3 =
\emptyset$ (i.e., vertex $v_b$ is close to the intersection point of
$(v_0,v_a)$ and $(v_b,v_c)$).

\begin{figure}[htb]
    \centering
    \includegraphics[width=.5\textwidth]{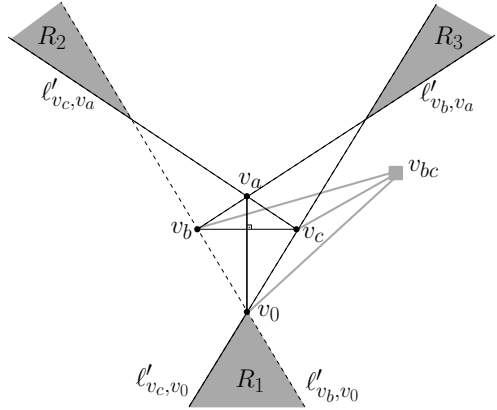}
    \caption{Vertex $v_{bc}$ should lie in the interior of $R_1$ or $R_2$ or $R_3$.}
    \label{fig:v0nocross-regions}
\end{figure}

\begin{description}

\item [Case i:] \emph{Vertex $v_{bc}$ is in the interior of $R_1$}. This case is
depicted in Figure~\ref{fig:v0nocross-a-1}. Let $T_{v_{bc}}$ be the
region formed by vertices $v_{bc}$, $v_b$ and $v_c$ (i.e., the
dark-gray colored region of Figure~\ref{fig:v0nocross-a-1}). Vertex
$v_d$, which has to be connected to vertices $v_a$ and $v_c$, and,
the central vertex $v_0$, cannot lie within $T_{v_{bc}}$, since edge
$(v_a,v_d)$ would have to cross edge $(v_b,v_c)$, which is already
involved in a right-angle crossing. Since vertex $v_d$ has to be
connected to vertex $v_0$, has to coincide with semi-line
$\ell_{v_0,v_c,v_{bc}}'$, as illustrated in
Figure~\ref{fig:v0nocross-a-1}. However, under this restriction, the
common neighbor $v_{cd}$ of vertices $v_c$ and $v_d$ cannot be
connected to vertex $v_0$, since edge $(v_0,v_{cd})$ should be
perpendicular to one of the edges of $T_{v_{bc}}$, which cannot be
accomplished without introducing an edge overlap with edge
$(v_0,v_d)$.

\begin{figure}[htb]
  \centering
  \begin{minipage}[b]{.48\textwidth}
     \centering
     \subfloat[\label{fig:v0nocross-a-1}{}]
     {\includegraphics[width=.85\textwidth]{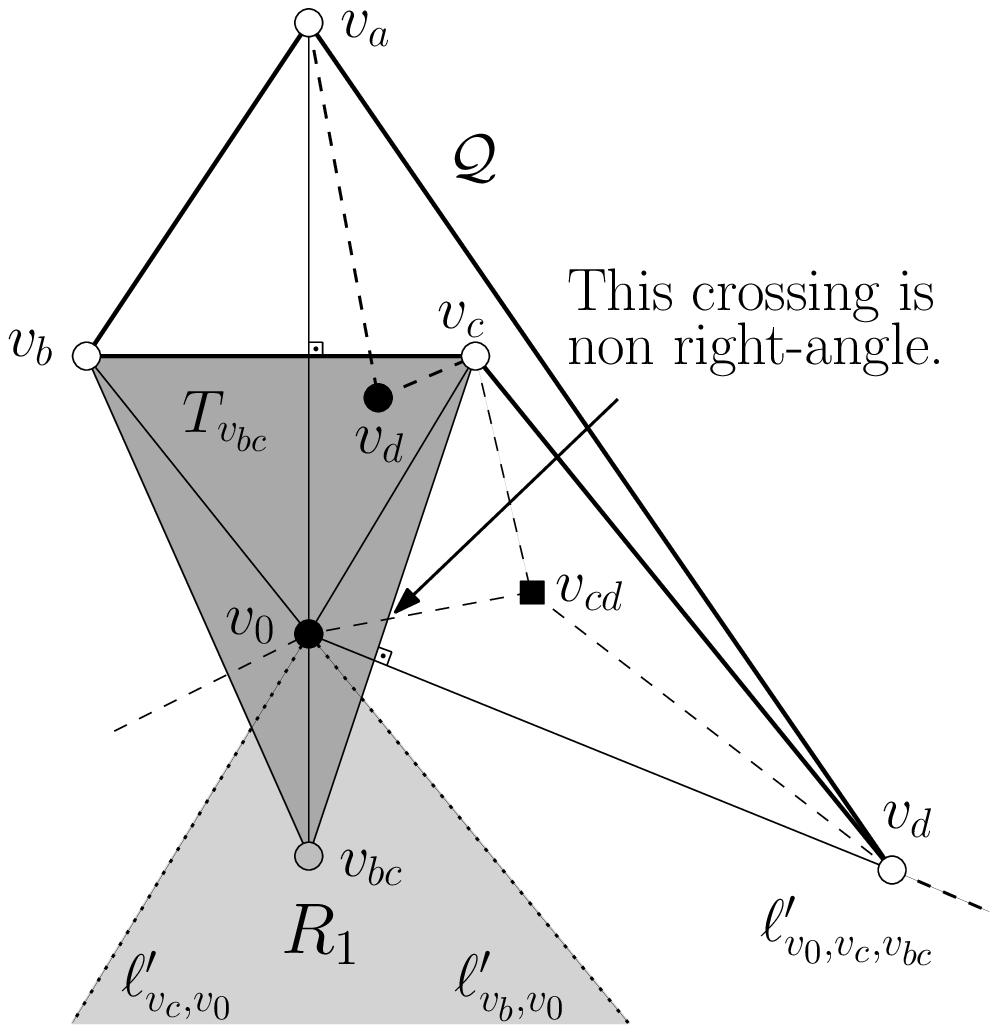}}
  \end{minipage}
  \hfill
  \begin{minipage}[b]{.48\textwidth}
     \centering
     \subfloat[\label{fig:v0nocross-b}{}]
     {\includegraphics[width=\textwidth]{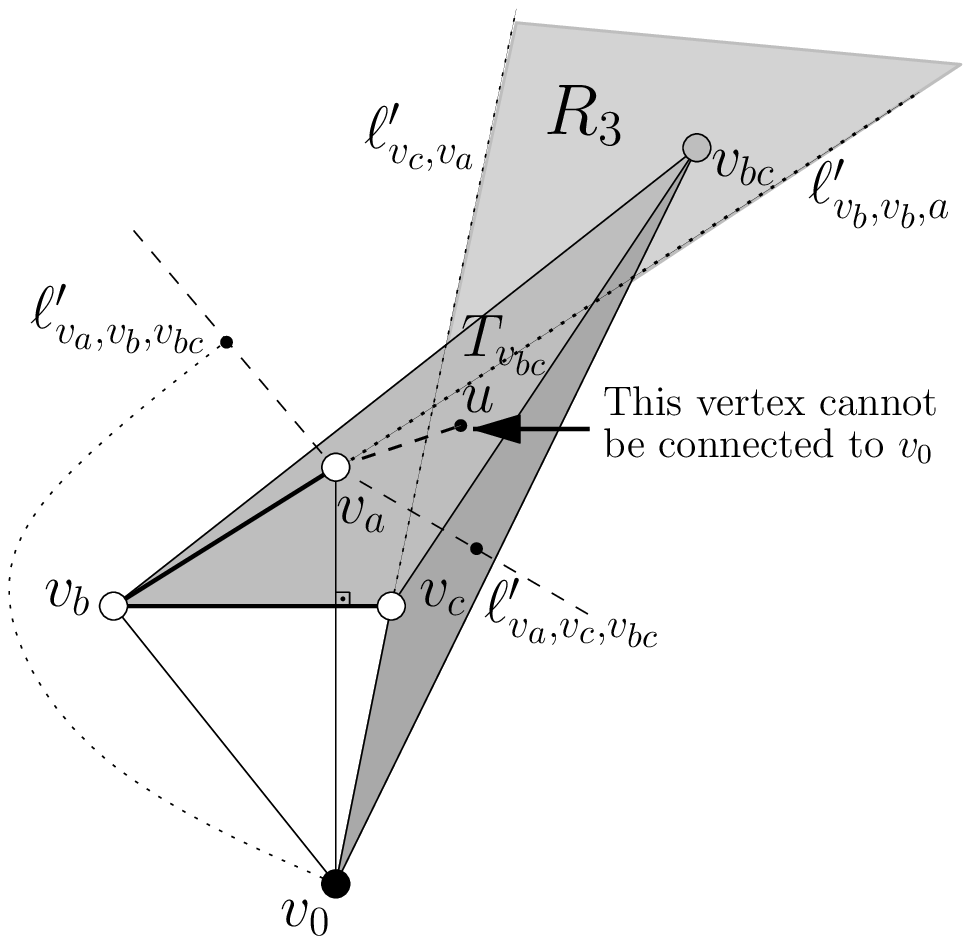}}
  \end{minipage}
  \caption{(a)~Vertex $v_{bc}$ lies in the interior of $R_1$. (b)~Vertex $v_{bc}$ lies in the interior of $R_3$.}
  \label{fig:v0nocross}
\end{figure}

\item [Case ii:] \emph{Vertex $v_{bc}$ is in the interior of either $R_2$ or $R_3$.}
Say without loss of generality that vertex $v_{bc}$ is in the
interior of $R_3$. This case is depicted in
Figure~\ref{fig:v0nocross-b}. Let $u$ be a vertex of the augmented
antiprism graph (except $v_a$) and assume that $u$ lies in the
interior of the triangular face, say $T_{v_{bc}}$, formed by
vertices $v_b$, $v_c$ and $v_{bc}$. Vertex $u$ has to be connected
to the central vertex $v_0$. Edge $(v_0,u)$ should not be involved
in crossings with edges $(v_0,v_a)$ and $(v_b,v_c)$, since they are
already involved in a right-angle crossing. If edge $(v_0,u)$
crosses edge $(v_0, v_{bc})$, then the three vertices $v_b$, $v_c$
and $v_{bc}$ that define triangle $T_{v_{bc}}$ must be collinear,
which leads to a contradiction. Therefore, triangle $T_{v_{bc}}$
cannot accommodate any other vertex (except $v_a$). Now observe that
each vertex of quadrilateral \qua~has degree five and there do not
exist three vertices of quadrilateral \qua, that have a common
neighbor (see Figure~\ref{fig:basic-gadget-1}). These properties
trivially hold for vertex $v_a$, since $v_a\in \mathcal{Q}$. Based
on the above properties, each neighbor of vertex $v_a$ can lie
either in the interior of the dark-gray region of
Figure~\ref{fig:v0nocross-b}, or, on the external face of the
already constructed drawing (along the dashed semi-lines
$\ell_{v_a,v_c,v_{bc}}'$ and $\ell_{v_a,v_b,v_{bc}}'$ of
Figure~\ref{fig:v0nocross-b}, respectively). This implies that we
can route only four vertices out of those incident to vertex $v_a$,
i.e., one of them should lie in the light-gray colored region of
Figure~\ref{fig:v0nocross-b} and thus, it cannot be connected to
vertex $v_0$. \qed

\end{description}
\end{proof}

\begin{lemma}
In any RAC drawing of the augmented square antiprism graph,
quadrilateral \quai, $i=1,2$ is drawn planar. \label{lem:p6Planar}
\end{lemma}

\begin{proof}
Let \qua~be one of quadrilaterals \quai, $i=1,2$, and let, as in the
previous lemma, $v_a$, $v_b$, $v_c$ and $v_d$ be its vertices,
consecutive along quadrilateral \qua. Assume to the contrary that in
a RAC drawing of the augmented square antiprism graph, quadrilateral
\qua~is not drawn planar, and say that edges $(v_a,v_b)$ and
$(v_c,v_d)$ form a right-angle crossing. This case is illustrated in
Figure~\ref{fig:q4planar}. In the following, we will lead to a
contradiction the cases, where central vertex $v_0$ lies (i)~in the
interior of a triangular face of quadrilateral \qua, and, (ii)~on
the external face of quadrilateral \qua. Note that it is not
feasible a non-planar RAC drawing of a quadrilateral to contain more
than one (right-angle) crossing. Hence, its bounded faces are
triangular.

\begin{itemize}

\item \textbf{Case i:} \emph{Vertex $v_0$ lies in the interior of a triangular face of quadrilateral \qua.}
Assume without loss of generality that vertex $v_0$ (which is
incident to all vertices of quadrilateral \qua) lies in the interior
of the triangular face formed by vertices $v_b$, $v_c$ and the
intersection point of edges $(v_a,v_b)$ and $(v_c,v_d)$, as in
Figure~\ref{fig:q4planar-2}. In this case, edges $(v_0,v_a)$,
$(v_a,v_b)$ and $(v_c,v_d)$ mutually cross, which leads to a
contradiction, due to Property~\ref{prp:three-crossing-edges}.

\begin{figure}[h!tb]
  \begin{minipage}[b]{.44\textwidth}
    \subfloat[\label{fig:q4planar-2}{Vertex $v_0$ lies in the interior of a triangular face.}]
     {\includegraphics[width=\textwidth]{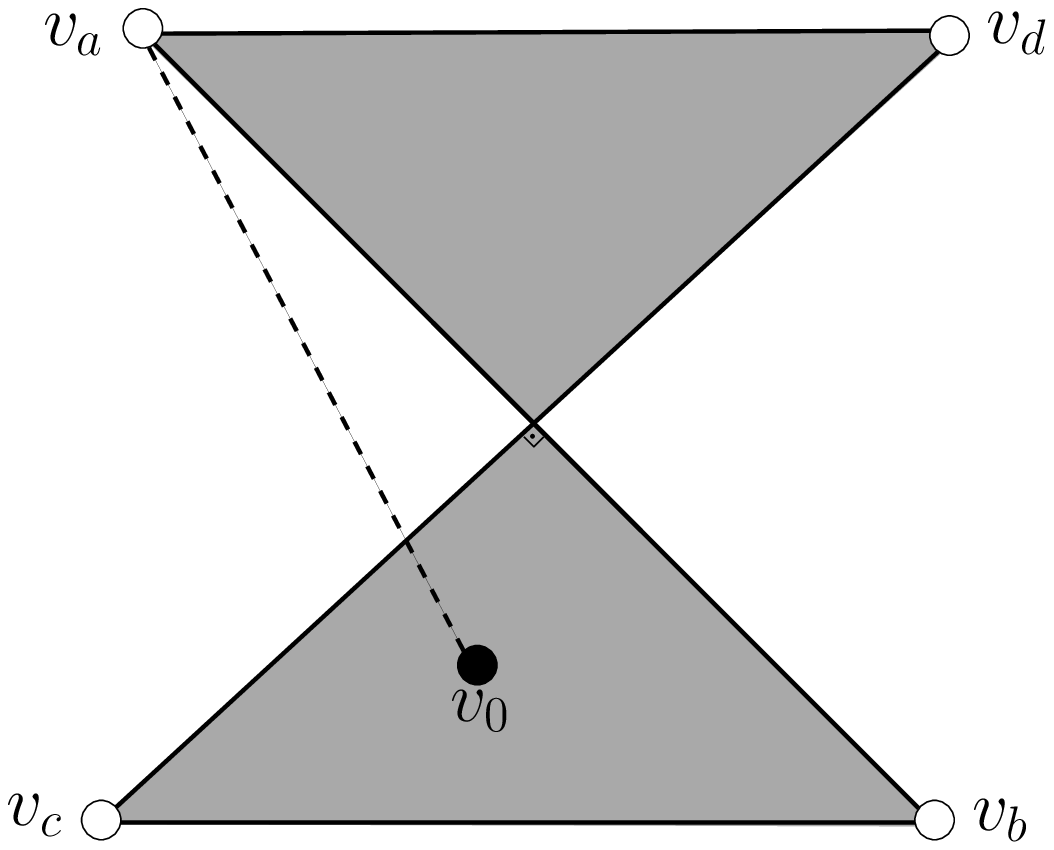}}
  \end{minipage}
  \hfill
  \begin{minipage}[b]{.52\textwidth}
    \centering
    \subfloat[\label{fig:q4planar-3}{Vertex $v_0$ lies on the external face of quadrilateral \qua.}]
    {\includegraphics[width=\textwidth]{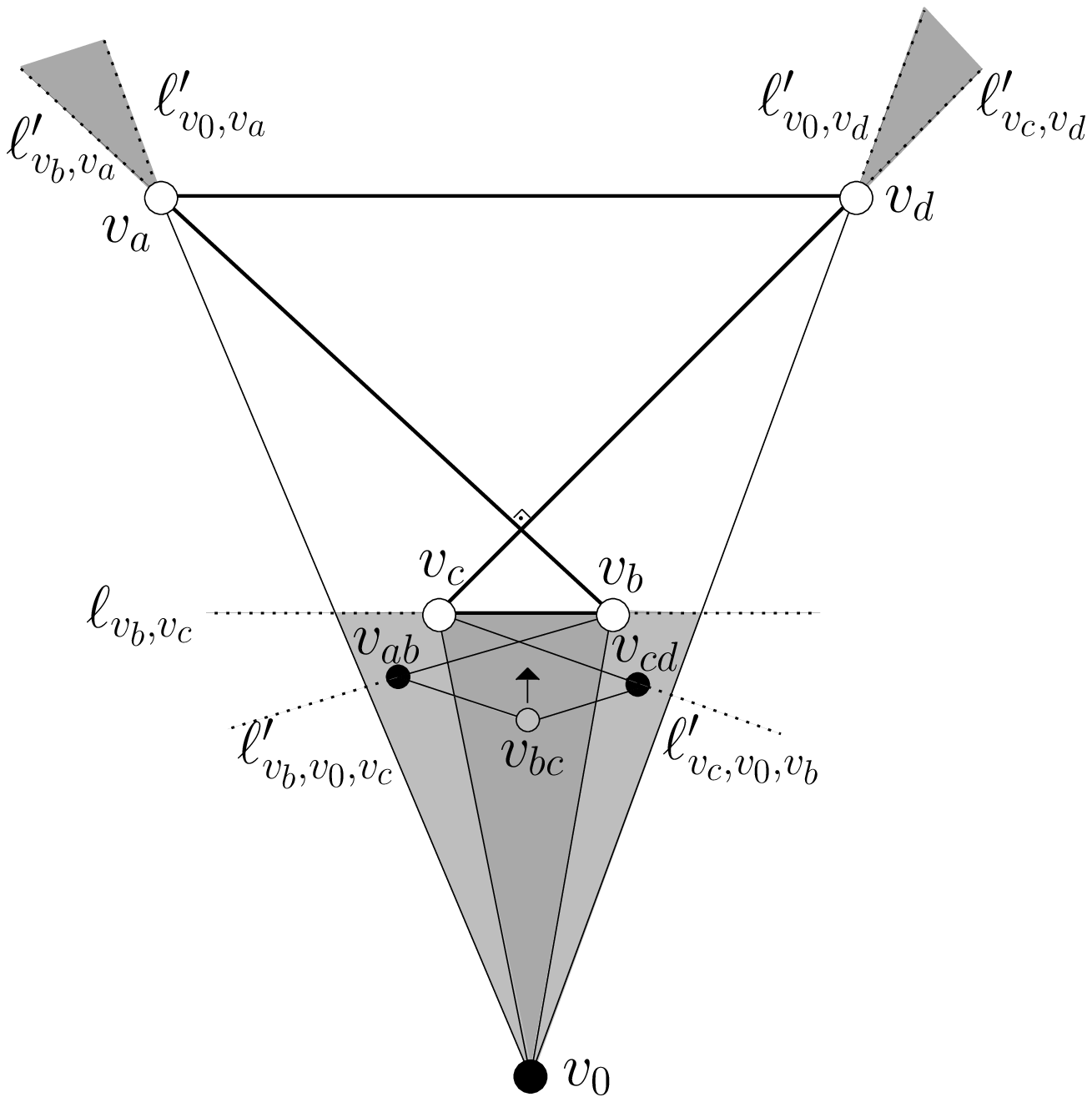}}
  \end{minipage}
  \caption{Quadrilateral \qua~is not drawn planar.}
  \label{fig:q4planar}
\end{figure}

\item \textbf{Case ii:} \emph{Vertex $v_0$ lies on the external face of
quadrilateral \qua.} This case is illustrated in
Figure~\ref{fig:q4planar-3}. Recall that by
Lemma~\ref{lem:v0nocross}, vertex $v_0$ cannot introduce additional
crossings on quadrilateral \qua. We will first show that the common
neighbor $v_{ab}$ of vertices $v_a$ and $v_b$ cannot lie in the
region ``above'' line $\ell_{v_b,v_c}$. In the case, where vertex
$v_{ab}$ lies in the region ``above'' $\ell_{v_b,v_c}$ and to the
``left'' of both edge $(v_c,v_d)$ and semi-line $\ell_{v_0,v_d}'$,
edge $(v_b,v_{ab})$ would cross edge $(v_c,v_d)$, which is not
permitted by Property~\ref{prp:three-crossing-edges}. Symmetrically,
vertex $v_{ab}$ cannot lie in the region ``above'' $\ell_{v_b,v_c}$
and to the ``right'' of both edge $(v_a,v_b)$ and semi-line
$\ell_{v_0,v_a}'$. If vertex $v_{ab}$ lies within the left
gray-colored unbounded region of Figure~\ref{fig:q4planar-3} (that
is formed by semi-lines $\ell_{v_0,v_a}'$, $\ell_{v_b,v_a}'$), then,
edge $(v_{ab}, v_b)$ crosses two non-parallel edges $(v_a,v_d)$ and
$(v_a,v_d)$. In the case where, $v_{ab}$ lies in the right
gray-colored unbounded region of Figure~\ref{fig:q4planar-3} (that
is formed by semi-lines $\ell_{v_0,v_d}'$, $\ell_{v_c,v_d}'$), then
$(v_a,v_{ab})$ either crosses both $(v_a,v_d)$ and $(v_c,v_d)$ which
are non-parallel, or crosses edge $(v_0,v_d)$ forming a non-right
angle crossing. In the case, where vertex $v_{ab}$ lies in the
interior of the triangle formed by vertices $v_b$, $v_c$ and the
intersection point of edges $(v_a,v_b)$ and $(v_c,v_d)$, edge
$(v_a,v_{ab})$ would cross edge $(v_c,v_d)$, which leads to a
violation Property~\ref{prp:three-crossing-edges}. Therefore, vertex
$v_{ab}$ should be ``below'' $\ell_{v_b,v_c}$.

We continue our reasoning on vertex $v_{ab}$. Vertex $v_{ab}$ cannot
lie to the ``left'' of edge $(v_0,v_a)$, since edge $(v_b,v_{ab})$
or $(v_a,v_{ab})$ would cross more than one (non-parallel) edges
incident to vertex $v_0$. If vertex $v_{ab}$ lies to the ``right''
of edge $(v_0,v_b)$, then edge $(v_a,v_{ab})$ either crosses edge
$(v_c,v_d)$, that it is not permitted by
Property~\ref{prp:three-crossing-edges}, or, both edges $(v_0,v_b)$
and $(v_0,v_c)$, that are non-parallel. We complete our reasoning on
vertex $v_{ab}$ by the triangle formed by vertices $v_0$, $v_b$ and
$v_c$. In this case, $(v_a,v_{ab})$ should be perpendicular to edge
$(v_0, v_c)$. This suggests that the angle formed by edges
$(v_c,v_d)$ and $(v_0,v_c)$ is greater that $180^o$ and therefore,
edge $(v_0,v_d)$ either crosses $(v_a,v_b)$, or another edge of
quadrilateral \qua, which trivially leads to a contradiction, due to
Property~\ref{prp:three-crossing-edges}, or due to
Lemma~\ref{lem:v0nocross}, respectively. Based on the above, vertex
$v_{ab}$ should be within the left gray-colored region of
Figure~\ref{fig:q4planar-3}, along semi-line $\ell_{v_b,v_c,v_0}'$.
Following a similar reasoning scheme, as for vertex $v_{ab}$, we can
prove that the common neighbor $v_{cd}$ of vertices $v_c$ and $v_d$
should lie within the right light-gray colored region of
Figure~\ref{fig:q4planar-3}, along semi-line $\ell_{v_c,v_0,v_b}'$.
However, in this case, a common neighbor of vertices $v_{ab}$ and
$v_{cd}$, say $v_{bc}$, should lie on the intersection of semi-lines
$\ell_{v_b,v_0,v_c}'$ and $\ell_{v_c,v_0,v_b}'$, which leads to edge
overlaps. Thus, there exists no feasible placement for vertex
$v_{bc}$. \qed
\end{itemize}
\end{proof}

\begin{lemma}
In any RAC drawing of the augmented square antiprism graph, the
central vertex $v_0$ lies in the interior of quadrilateral \quai,
$i=1,2$. \label{lem:v0internal}
\end{lemma}

\begin{proof}

From Lemma~\ref{lem:p6Planar}, it follows that quadrilateral
\quai~should be drawn planar, for each $i=1,2$. In order to prove
this lemma, we assume to the contrary that central vertex $v_0$ lies
on the exterior of one of the two quadrilaterals. Say, w.l.o.g., on
the exterior of quadrilateral \quaa. Let $v_a$, $v_b$, $v_c$ and
$v_d$ be \quaa's vertices, consecutive along quadrilateral \quaa.
Then, by Lemma~\ref{lem:v0nocross}, vertex $v_0$ cannot contribute
additional crossings on quadrilateral \quaa. This suggests that the
drawing of the graph induced by quadrilateral \quaa~and vertex $v_0$
will resemble the one depicted in Figure~\ref{fig:v0internal-gen}.
We denote by $T_{v_0}$ the triangle formed by vertex $v_0$ and the
two vertices, which are on the convex hall of \quaa$\cup v_0$~(refer
to the gray-colored triangle of Figure~\ref{fig:v0internal-gen}).

\begin{figure}[htb]
  \centering
  \includegraphics[width=.35\textwidth]{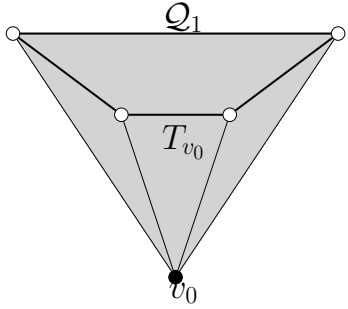}
  \caption{Any drawing of the graph induced by \quaa~and $v_0$ has to resemble to this one.}
  \label{fig:v0internal-gen}
\end{figure}

Before we proceed with the detailed proof of this lemma, we recall
some properties of the augmented square antiprism graph. Two
consecutive vertices of \quaa~(\quab)~share a common vertex of
quadrilateral \quab~(\quaa). Each vertex of quadrilateral
\quaa~(\quab)~should be connected to two consecutive vertices of
quadrilateral \quab~(\quaa). We will prove that (i)~no vertex of
\quab~lies outside $T_{v_0}$, (ii)~\quab~cannot entirely lie in the
interior of \quaa, (iii)~\quab~cannot entirely lie in the interior
of a triangular face of $T_{v_0}$, (iv)~\quab~cannot entirely lie
within two adjacent triangular faces of $T_{v_0}$, (v)~\quab~cannot
cross \quaa, such that some of the vertices of \quab~reside within a
triangular face of $T_{v_0}$, whereas the remaining ones within
\quaa. Note that quadrilateral \quab~cannot entirely lie within
three triangular faces of $T_{v_0}$ incident to vertex $v_0$.

\begin{description}

\item [Case i:] We prove that no vertex of quadrilateral \quab~lies on
the external face of the graph induced by quadrilateral \quaa~and
vertex $v_0$, i.e., outside $T_{v_0}$. For the sake of
contradiction, assume that there exists a vertex of quadrilateral
\quab, say $v_{ab}$, that lies on the external face of the graph
induced by quadrilateral \quaa~and vertex $v_0$ (see
Figure~\ref{fig:q4external}). Vertex $v_{ab}$ should be connected to
vertices $v_a$ and $v_b$ of quadrilateral \quaa, and to the central
vertex $v_0$. If both vertices $v_a$ and $v_b$ are inside triangle
$T_{v_0}$, then vertex $v_{ab}$, which is assumed to lie on the
external face of this graph, would violate
Property~\ref{prp:triangle-edges}, since vertices $v_a$ and $v_b$
would lie in the interior of $T_{v_0}$, whereas vertex $v_{ab}$
outside. Therefore, at least one of vertices $v_a$ and $v_b$ should
be a corner of $T_{v_0}$. Then, vertex $v_{ab}$ contributes either
none (see Figure~\ref{fig:q4external-2}), or a single right-angle
crossing (see Figure~\ref{fig:q4external-1}).

\begin{figure}[h!tb]
  \centering
  \begin{minipage}[b]{.46\textwidth}
     \centering
     \subfloat[\label{fig:q4external-1}{}]
     {\includegraphics[width=\textwidth]{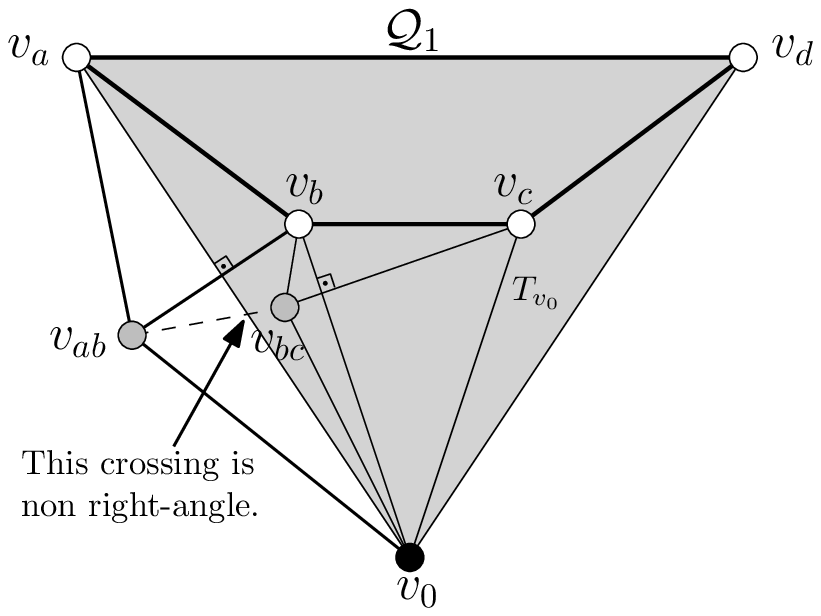}}
   \end{minipage}
  \hfill
  \begin{minipage}[b]{.50\textwidth}
     \centering
     \subfloat[\label{fig:q4external-2}{}]
     {\includegraphics[width=\textwidth]{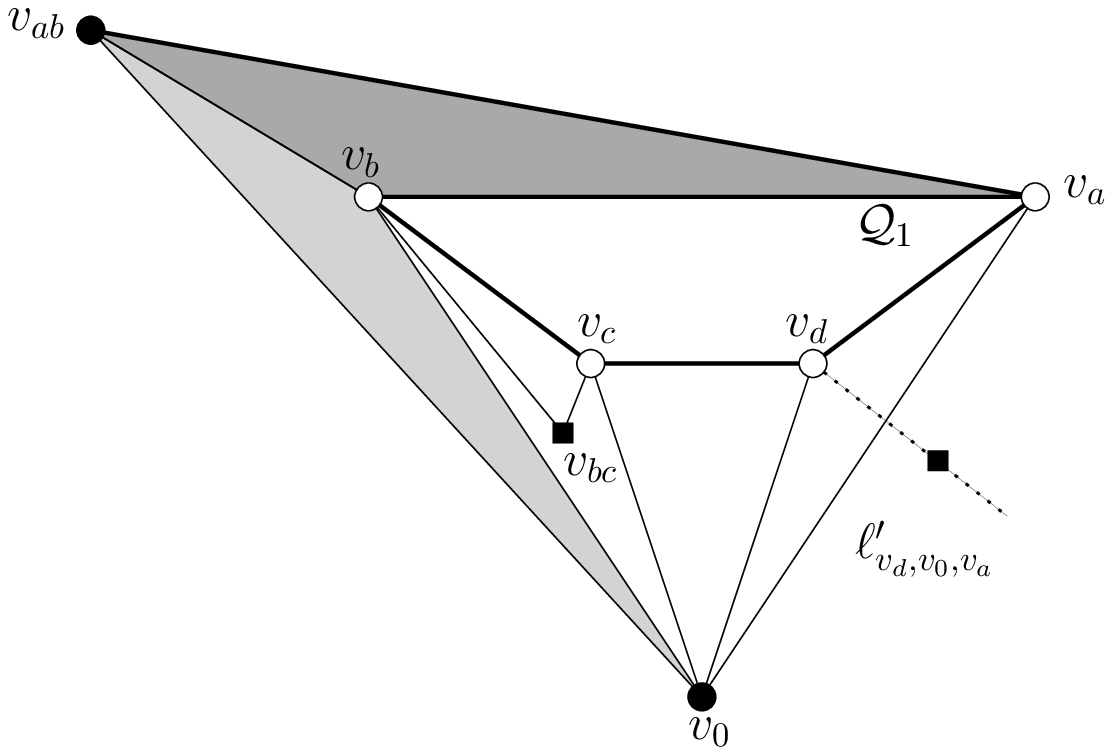}}
  \end{minipage}
  \begin{minipage}[b]{.55\textwidth}~\end{minipage}
  \caption{Vertex $v_{ab}$ of quadrilateral \quab~lies on the external face of the graph induced by quadrilateral \quaa~and vertex $v_0$.}
  \label{fig:q4external}
\end{figure}

Let now $v_{bc}$ be a vertex of quadrilateral \quab, which is
incident to vertex $v_{ab}$. Vertex $v_{bc}$ is also incident to two
consecutive vertices of quadrilateral \quaa, i.e., $v_b$ and $v_c$.
We first turn our attention in the case where $v_{ab}$ contributes a
single right-angle crossing on the graph induced by quadrilateral
\quaa~and vertex $v_0$ (see Figure~\ref{fig:q4external-1}). Then, by
Property~\ref{prp:triangle-edges}, vertex $v_{bc}$ should lie in the
interior of triangle $T_{v_0}$. This immediately leads to a
contradiction, since edge $(v_{ab},v_{bc})$ should cross edge
$(v_0,v_a)$, which is already involved in a right-angle crossing
(see Figure~\ref{fig:q4external-1}).

Consider now the case where vertex $v_{ab}$ contributes no crossing
on $T_{v_0}$. Observe that vertex $v_{bc}$, which is adjacent to
vertex $v_{ab}$, vertices $v_b$ and $v_c$ of \quaa, and the central
vertex $v_0$ cannot lie in the dark-gray region of
Figure~\ref{fig:q4external-2}, since in this case, edge $(v_a,v_b)$
would be crossed by more than one (non-parallel) edges, incident to
$v_{bc}$. The case where vertex $v_{bc}$ lies within the light-gray
colored region of Figure~\ref{fig:q4external-2}, leads to a
situation similar to the one depicted in
Figure~\ref{fig:q4external-1}. Therefore, vertex $v_{bc}$ should lie
``somewhere'' in the interior of $T_{v_0}$. Let $v_{ad}$ be the
common neighbor of vertices $v_a$ and $v_d$ of \quaa, and vertex
$v_{ab}$. This vertex cannot lie within the dark-gray region of
Figure~\ref{fig:q4external-2}, for the same reason that vertex
$v_{bc}$ couldn't. In addition, vertex $v_{ad}$ cannot lie in the
interior of $T_{v_0}$, since in this case, both vertices $v_{bc}$
and $v_{ad}$ (that are in $T_{v_0}$), should be connected to
$v_{ab}$ (that is not in $T_{v_0}$), which trivially violates
Property~\ref{prp:triangle-edges}. Therefore, vertex $v_{ad}$ should
be on the external face of the graph induced by \quaa~and vertices
$v_0$ and $v_{ad}$, along semi-line $\ell_{v_d,v_0,v_a}'$. However,
in this case, we are also led to a situation similar to the one
depicted in Figure~\ref{fig:q4external-1}, and subsequently, to a
contradiction.

\item [Case ii:] Say that quadrilateral
\quab~entirely lies within quadrilateral \quaa~(see
Figure~\ref{fig:q4internal-1}). In this case, its vertices should be
connected to vertex $v_0$. For three vertices of quadrilateral
\quab, this can be accomplished using the three available edges of
quadrilateral \quaa~(refer to the dotted edges of
Figure~\ref{fig:q4internal-1}), such that the right-angle crossings
occur along them. However, the fourth vertex cannot be connected to
vertex $v_0$, since only three edges of quadrilateral \quaa~can be
used to realize connections with vertex $v_0$ (see the topmost edge
of Figure~\ref{fig:q4internal-1}).

\begin{figure}[h!tb]
  \centering
  \begin{minipage}[b]{.30\textwidth}
     \centering
     \subfloat[\label{fig:q4internal-1}{}]
     {\includegraphics[width=\textwidth]{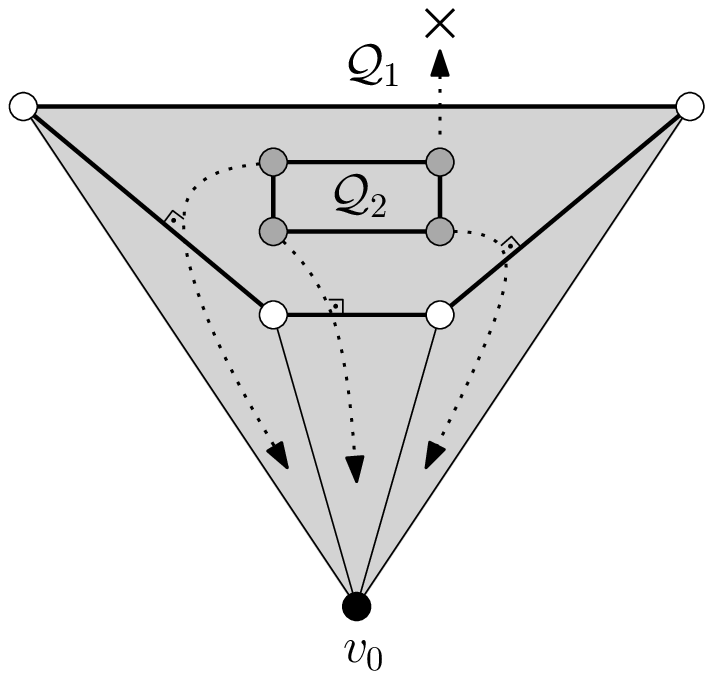}}
   \end{minipage}
  \hfill
  \begin{minipage}[b]{.33\textwidth}
     \centering
     \subfloat[\label{fig:q4internal-2}{}]
     {\includegraphics[width=\textwidth]{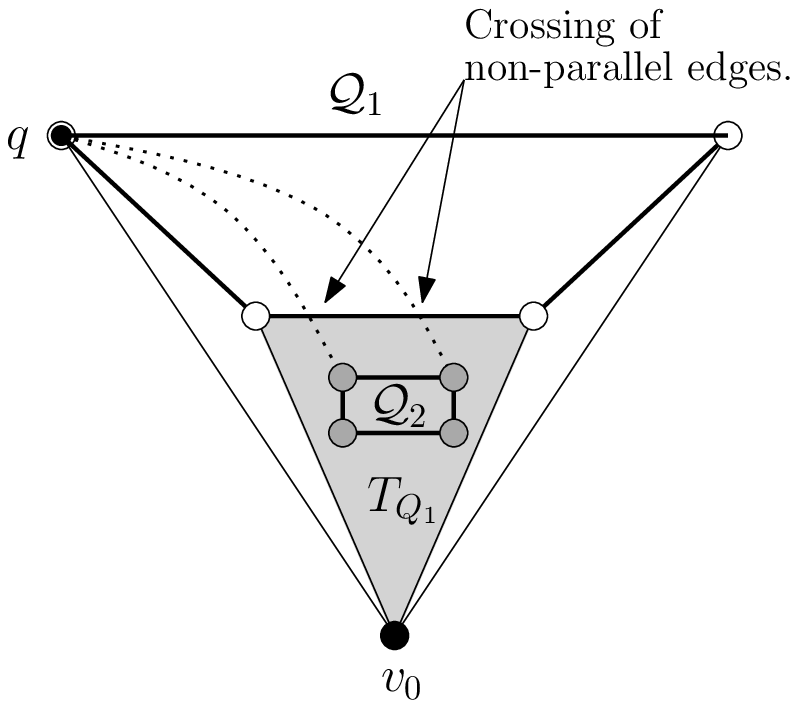}}
  \end{minipage}
  \hfill
  \begin{minipage}[b]{.33\textwidth}
     \centering
     \subfloat[\label{fig:q4internal-3}{}]
     {\includegraphics[width=\textwidth]{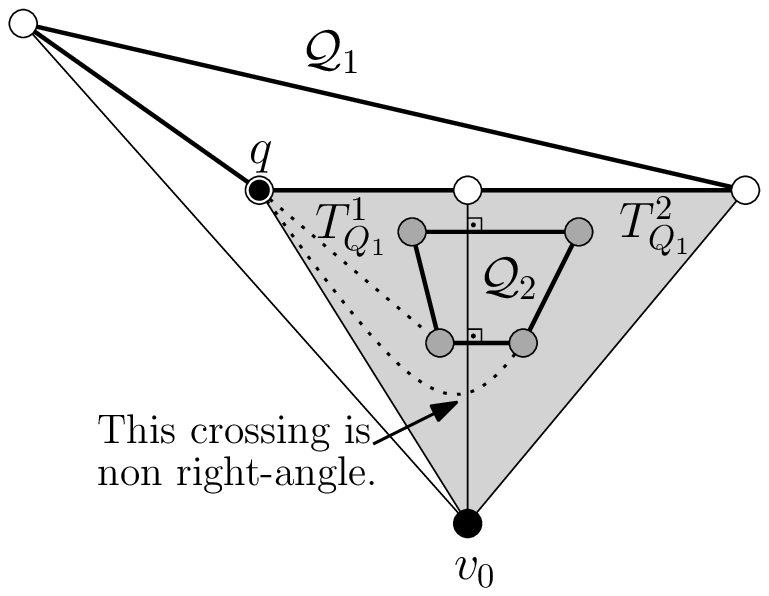}}
  \end{minipage}
  \caption{Quadrilateral \quab~lies (a)~in the interior of \quaa, (b)~in the interior of a triangular face
  $T_{Q_2}$ incident to vertex $v_0$, and (c)~within two adjacent triangular faces incident to vertex $v_0$.}
  \label{fig:q4internal-a}
\end{figure}

\item [Case iii:] Assume now that quadrilateral \quab~entirely lies within a
triangular face, say $T_{Q_1}$, incident to vertex $v_0$ (see
Figure~\ref{fig:q4internal-2}). Then, there exists at least one
vertex of quadrilateral \quaa, say vertex $q$, which is incident to
two vertices of quadrilateral \quab, and is not identified with a
vertex at the corners of $T_{Q_1}$ (see
Figure~\ref{fig:q4internal-2}). Vertex $q$ has to be connected to
two vertices of quadrilateral \quab. However, vertex $q$ is external
to triangle $T_{Q_1}$, whereas its two incident vertices in the
interior of this triangle, which leads to a violation of
Property~\ref{prp:triangle-edges}.

\item [Case iv:] Say that quadrilateral \quab~entirely lies within two adjacent
triangular faces, say $T_{Q_1}^1$ and $T_{Q_1}^2$, incident to
vertex $v_0$ (see Figure~\ref{fig:q4internal-3}). Then,
quadrilateral \quab~should be ``perpendicular'' to the common edge
of $T_{Q_1}^1$ and $T_{Q_1}^2$. Recall that two consecutive vertices
of quadrilateral \quab~share a common vertex of quadrilateral \quaa.
Hence, we can find a vertex $q$ of quadrilateral \quaa, which is not
identified with the common vertex of $T_{Q_1}^1$ and $T_{Q_1}^2$,
and is incident to a pair of vertices of quadrilateral \quab, that
do not lie in the same triangular face (i.e., the topmost vertices
of quadrilateral \quab~or the bottommost vertices of quadrilateral
\quab~in Figure~\ref{fig:q4internal-3}). This leads to a
contradiction, since the common edge of $T_{Q_1}^1$ and $T_{Q_4}^2$
cannot be crossed, as it is already involved in a right-angle
crossing (refer to the dotted-edges of
Figure~\ref{fig:q4internal-3}).

\item [Case v:] We consider the case where quadrilateral \quab~crosses quadrilateral
\quaa, such that some of the vertices of quadrilateral \quab~reside
within a triangular face of $T_{v_0}$, whereas the remaining ones
within quadrilateral \quaa. We will lead to a contradiction the
cases where: (i)~Two vertices of \quab~lie in the interior of a
single triangular face incident to $v_0$, (ii)~two vertices of
\quab~lie in the interior of two adjacent triangular faces,
(iii)~three vertices of \quab~lie in the interior of two adjacent
triangular faces and two of them lie in the same triangular face of
$T_{v_0}$, (iv)~three vertices of \quab~lie in the interior of three
pairwise-adjacent triangular faces incident to vertex $v_0$. Recall
that none of the vertices of \quab~lies in the external face of the
graph induced by quadrilateral \quaa~and vertex $v_0$. Let, with a
slight abuse of notation, $q_a$, $q_b$, $q_c$ and $q_d$ be the
vertices of quadrilateral \quab. Assume first that vertices $q_a$
and $q_b$ are in the interior of a single triangular face, whereas
vertices $q_c$ and $q_d$ in the interior of quadrilateral \quaa~(see
Figure~\ref{fig:q4internal-4}). In this case, edges $(q_a,q_d)$ and
$(q_b,q_c)$ should perpendicularly cross quadrilateral \quaa. The
connections between vertices $q_c$ and $q_d$ with vertex $v_0$ can
be accomplished using two of the available edges of quadrilateral
\quaa, such that the right-angle crossings occur along them (refer
to dotted edges of Figure~\ref{fig:q4internal-4}). Thus, the
triangular faces that are adjacent to the one that accommodates
vertices $q_a$ and $q_b$ (refer to the light-gray faces of
Figure~\ref{fig:q4internal-4}) cannot be further used to connect
vertices of quadrilateral \quaa~to vertices of quadrilateral \quab.
Then, there exists a vertex of quadrilateral \quaa, say $q$, that it
is not identified with any of the vertices of the face that
accommodates vertices $q_a$ and $q_b$, and either $q_a$ or $q_b$ has
to be connected to vertex $q$. However, this cannot be accomplished,
since the edge from either $q_a$ or $q_b$ to vertex $q$ would cross
more than one non-parallel edges (refer to the dashed edge of
Figure~\ref{fig:q4internal-4}).

\begin{figure}[htb]
  \centering
  \includegraphics[width=.55\textwidth]{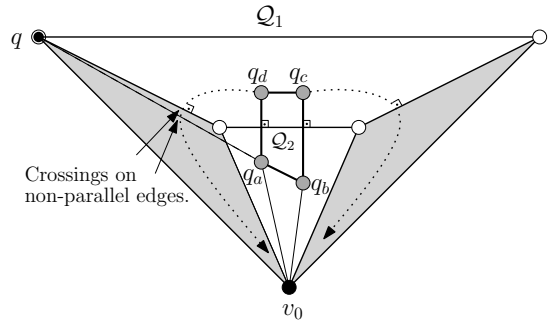}
  \caption{Vertices $q_a$ and $q_b$ are in the interior of a single triangular face incident to $v_0$.}
  \label{fig:q4internal-4}
\end{figure}

Say now that vertices $q_a$ and $q_b$ are in the interior of two
adjacent triangular faces incident to vertex $v_0$, whereas vertices
$q_c$ and $q_b$ within quadrilateral \quaa. This case is illustrated
in Figure~\ref{fig:q4internal-5}. Then, one of the vertices that lie
in the interior of quadrilateral \quaa, say vertex $q_d$, can be
connected to vertex $v_0$ using one of the available edges of
quadrilateral \quaa~(refer to the dotted edge of
Figure~\ref{fig:q4internal-5}). However, vertex $q_c$ cannot be
connected to vertex $v_0$, since only three of the edges of
quadrilateral \quaa~can be used to realize connections from the
vertices that lie within quadrilateral \quaa, to vertex $v_0$.

Consider now the case where three vertices, say $q_a$, $q_b$ and
$q_c$ of quadrilateral \quab~are in the interior of two adjacent
triangular faces incident to vertex $v_0$, and two of vertices
$q_a$, $q_b$ and $q_c$, say w.l.o.g., $q_b$ and $q_c$, lie in the
same triangular face (see Figure~\ref{fig:q4internal-7}). Then,
vertex $q_d$, as in the previous case, can be connected to vertex
$v_0$ using the ``last'' available edge of quadrilateral
\quaa~(refer to the dotted edge of Figure~\ref{fig:q4internal-7}).
However, in this case, there exists a vertex of quadrilateral \quaa,
say $q$, that it is not identified with any of the vertices of the
face that accommodates vertices $q_a$, $q_b$ and $q_c$, which has to
be connected to one of the vertices $q_a$, $q_b$ or $q_c$. However,
this cannot be accomplished, since an edge from either vertex $q_a$,
or $q_b$, or $q_c$, to vertex $q$ would cross more than one
non-parallel edges (refer to the dashed edge of
Figure~\ref{fig:q4internal-7}).

\begin{figure}[h!tb]
  \centering
  \begin{minipage}[b]{.28\textwidth}
     \centering
     \subfloat[\label{fig:q4internal-5}{}]
     {\includegraphics[width=\textwidth]{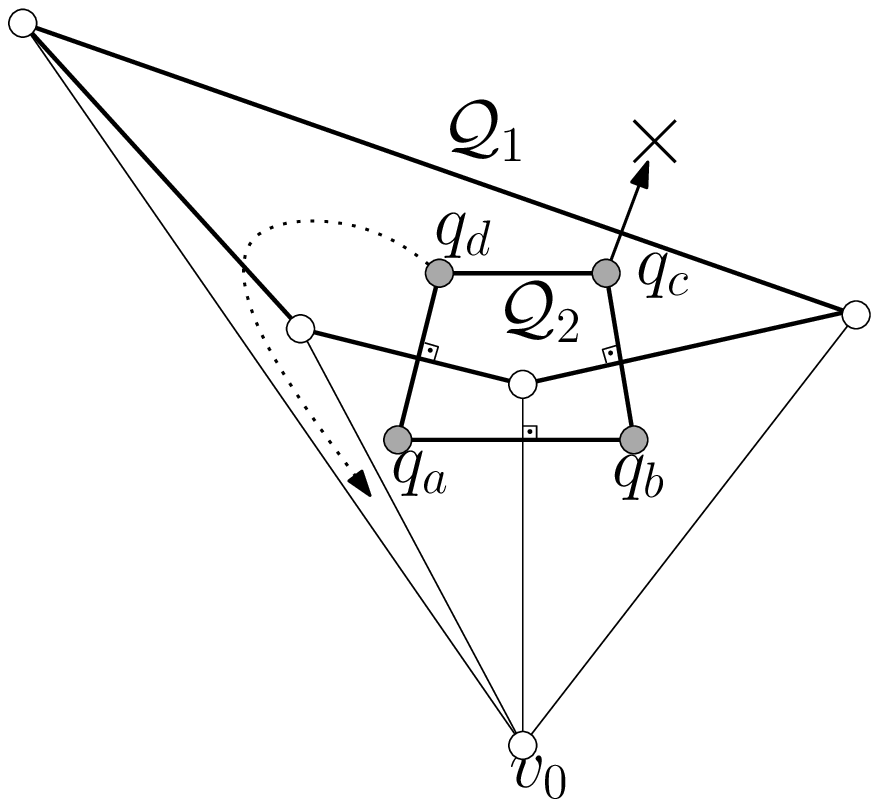}}
   \end{minipage}
   \hfill
  \begin{minipage}[b]{.29\textwidth}
     \centering
     \subfloat[\label{fig:q4internal-7}{}]
     {\includegraphics[width=\textwidth]{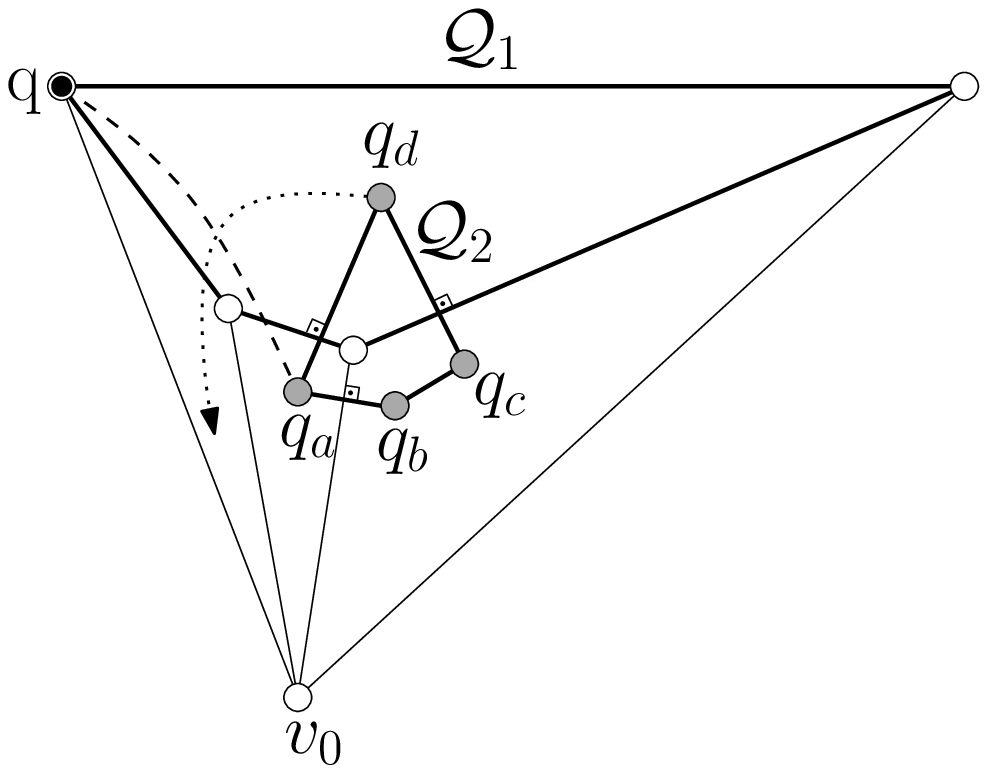}}
  \end{minipage}
  \hfill
  \begin{minipage}[b]{.39\textwidth}
     \centering
     \subfloat[\label{fig:q4internal-6}{}]
     {\includegraphics[width=\textwidth]{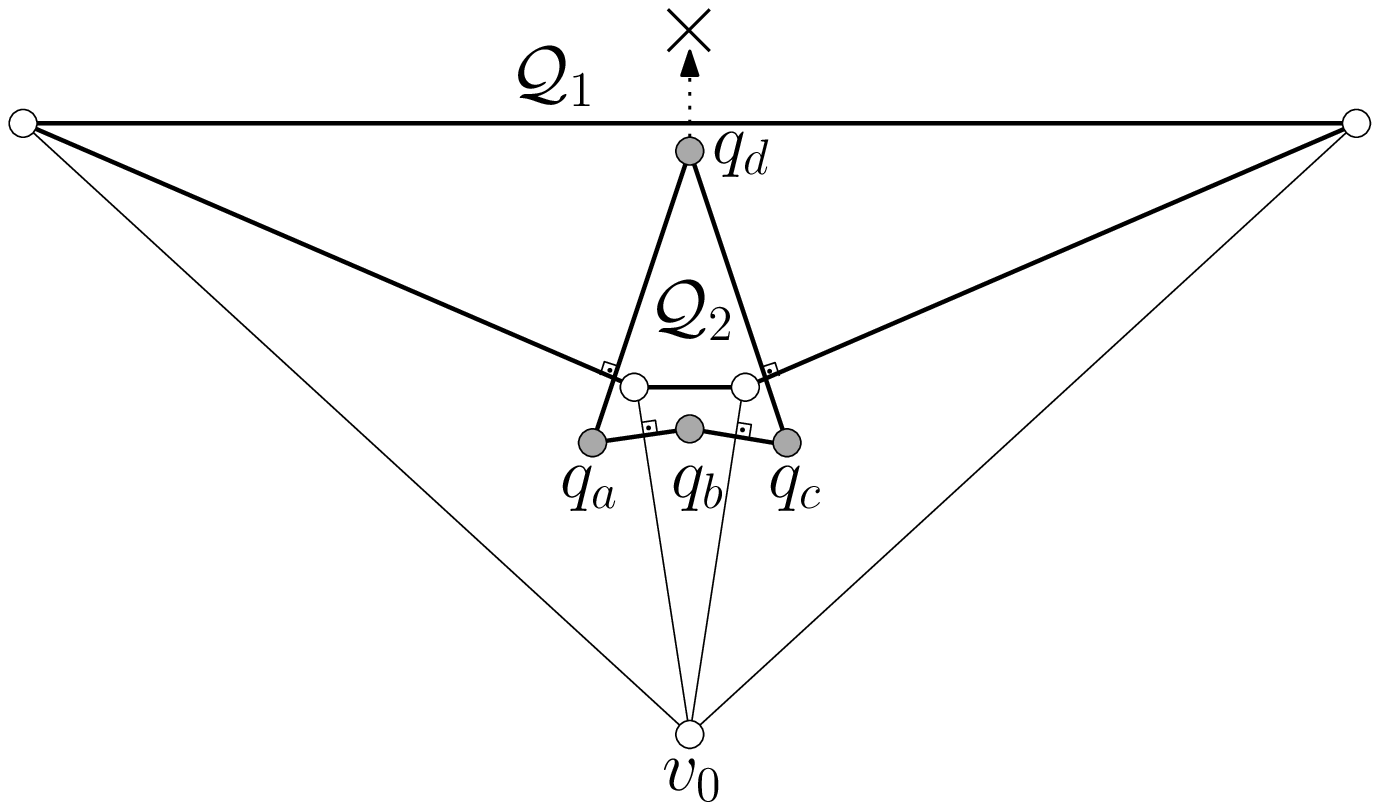}}
  \end{minipage}

  \caption{(a)~Vertices $q_a$ and $q_b$ are in the interior of two adjacent triangular faces incident to vertex $v_0$. Vertices $q_a$ and $q_b$ are not in the same
  triangular face incident to vertex $v_0$. (b)~Vertices $q_a$, $q_b$ and $q_c$ are in the interior of two adjacent triangular faces incident to vertex $v_0$, and $q_b$
     and $q_c$ lie in the same triangular face. (c)~Vertices $q_a$, $q_b$ and $q_c$ are in the interior of three pairwise-adjacent triangular faces incident to vertex $v_0$.}
  \label{fig:q4internal-d}
\end{figure}

The last case we have to consider is the one where three vertices,
say $q_a$, $q_b$ and $q_c$ of quadrilateral \quab~are in the
interior of three pairwise-adjacent triangular faces incident to
vertex $v_0$, whereas the fourth vertex $q_d$ resides within
quadrilateral \quab~(see Figure~\ref{fig:q4internal-6}). In this
case, vertex $q_d$ has to use the fourth edge of quadrilateral
\quaa~to reach vertex $v_0$, which leads to a contradiction, since
only three of the edges of quadrilateral \quaa~can be used to
realize connections from the vertices that lie within quadrilateral
\quaa, to vertex $v_0$.

\end{description}
Thus, we have considered all possible placements of \quab, with
vertex $v_0$ outside of \quaa, and are all led to a contradiction.
We conclude that vertex $v_0$ is in the interior of quadrilateral
\quaa~(and symmetrically in the interior of \quab, too). \qed
\end{proof}

\begin{lemma}
There does not exist a RAC drawing of the augmented square antiprism
graph where an edge emanating from vertex $v_0$ towards a vertex of
quadrilateral \quai, $i=1,2$, crosses quadrilateral \quai.
\label{lem:v0internalnocross}
\end{lemma}

\begin{proof}
By Lemma~\ref{lem:v0internal}, vertex $v_0$ should lie in the
interior of quadrilateral \quai, $i=1,2$, which is drawn planar due
to Lemma~\ref{lem:p6Planar}. Assume to the contrary that in a RAC
drawing of the augmented square antiprism graph, an edge emanating
from vertex $v_0$ towards a vertex of quadrilateral \quaa, say
$v_a$, crosses an edge, say $(v_c,v_d)$, of quadrilateral \quaa~(see
Figure~\ref{fig:planar-convex}). Consider vertex $v_{cd}$, which is
incident to vertices $v_c$ and $v_d$ of quadrilateral \quab. Vertex
$v_{cd}$ cannot lie ``above'' line $\ell_{v_c,v_d}$ and to the
``left'' of semi-line $\ell_{v_d,v_a}'$, since it cannot be
connected to vertex $v_0$. In addition, it cannot lie ``above'' line
$\ell_{v_c,v_d}$ and to the ``right'' of semi-line
$\ell_{v_d,v_a}'$, since in this case it cannot be connected to
either vertex $v_c$ or $v_0$. Furthermore, vertex $v_{cd}$ cannot be
in the interior of the triangle formed by vertices $v_0$, $v_c$ and
$v_d$, as it would not be feasible to be connected either to vertex
$v_c$ or $v_d$, since in either case, it crosses edge $(v_0,v_a)$.
Also, $v_{cd}$ cannot be in the region formed by line
$\ell_{v_c,v_d}$ and edges $(v_0,v_d)$ and $(v_0,v_b)$, as it could
not be connected to vertex $v_c$. Thus, vertex $v_{cd}$ should lie
in the light-gray triangular face of Figure~\ref{fig:planar-convex},
along semi-line $\ell_{v_d,v_c,v_b}$. Following a similar reasoning
scheme, we can prove that vertex $v_{ad}$, which is incident to
vertices $v_a$, $v_d$ of quadrilateral \quaa~and vertex $v_{ab}$ of
quadrilateral \quab, due to its adjacency with $v_a$, $v_d$, can lie
in the face formed by vertices $v_a$, $v_b$, $v_d$ and the
intersection point of edges $(v_d,v_{cd})$ and $(v_0,v_b)$. However,
under this restriction, vertex $v_{ad}$ cannot be connected to
vertex $v_{cd}$, without crossing edge $(v_0,v_b)$, which is already
involved in a right-angle crossing (refer to the dashed edge of
Figure~\ref{fig:planar-convex}). \qed
\end{proof}

 \begin{figure}[h!tb]
  \centering
  \includegraphics[width=.7\textwidth]{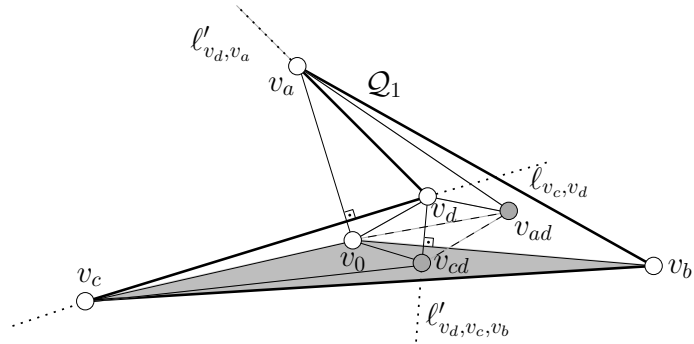}
  \caption{An edge emanating from vertex $v_0$ towards a vertex of \quaa, cannot cross
  \quaa.}
   \label{fig:planar-convex}
\end{figure}

\begin{lemma}
There does not exist a RAC drawing of the augmented square antiprism
graph in which quadrilateral \quaa~intersects \quab.
\label{lem:quahexposition}
\end{lemma}

\begin{proof}
From Lemmata~\ref{lem:p6Planar}, \ref{lem:v0internal} and
\ref{lem:v0internalnocross}, it follows that the graph induced by
quadrilateral \quaa~and vertex $v_0$ is drawn planar with vertex
$v_0$ in the interior of both quadrilaterals \quaa~and \quab.
Therefore, it should resemble the one illustrated in
Figure~\ref{fig:v0internalq4external-1}.

 \begin{figure}[h!tb]
  \centering
  \begin{minipage}[b]{.48\textwidth}
     \centering
     \subfloat[\label{fig:v0internalq4external-1}{}]
     {\includegraphics[width=.7\textwidth]{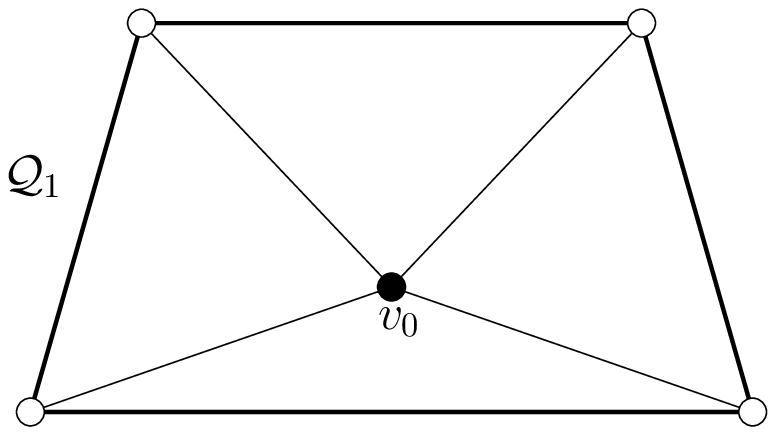}}
  \end{minipage}
   \hfill
  \begin{minipage}[b]{.48\textwidth}
     \centering
     \subfloat[\label{fig:v0internalq4external-6}{}]
     {\includegraphics[width=.7\textwidth]{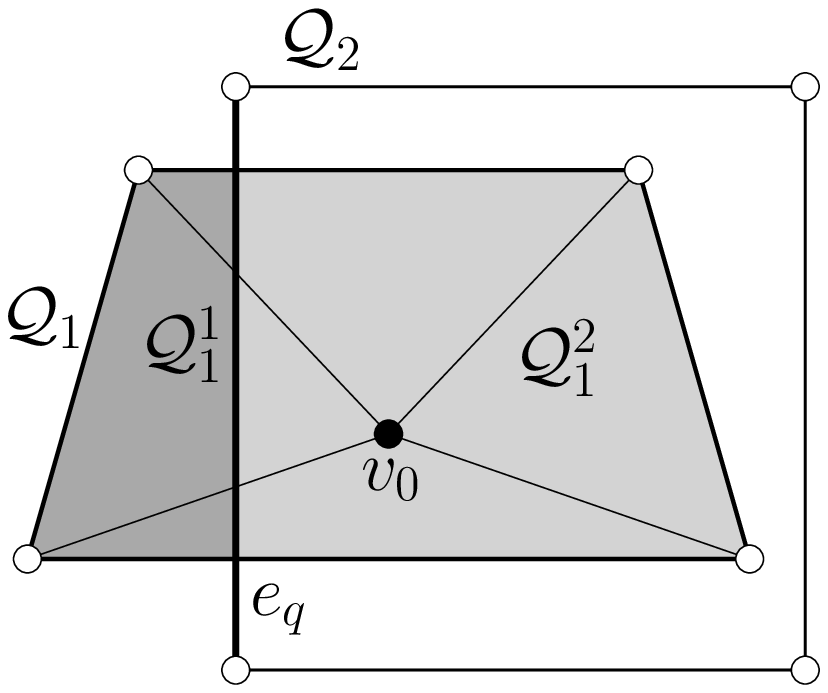}}
   \end{minipage}

  \caption{(a)~The graph induced by quadrilateral \quaa~and vertex $v_0$ is drawn
   planar with vertex $v_0$ in the interior of quadrilateral \quaa. (b)~\quaa~and \quab~cross and none of the vertices of \quab~is in the interior of \quaa.}
  \label{fig:v0internalq4external-a}
\end{figure}

In order to prove this lemma, we will contradict the following
cases: (i)~\quaa~and \quab~cross and none of the vertices of
\quab~is in the interior of quadrilateral \quaa, (ii)~two vertices
of \quab~lie in the interior of \quaa~and \quab~crosses either a
single edge of \quaa, or two edges of \quaa, (iii)~three vertices of
\quab~lie in the interior of \quaa, (iv)~only one vertex of
\quab~lies in the interior of \quaa. We first assume that
quadrilateral \quaa~and quadrilateral \quab~cross and none of the
vertices of quadrilateral \quab~is in the interior of quadrilateral
\quaa~(see Figure~\ref{fig:v0internalq4external-6}). In this case,
an edge of quadrilateral \quab, say $e_q$, which is involved in the
crossing, divides quadrilateral \quaa~into two regions, say
\quaaa~and \quaab. Obviously, edge $e_q$ should cross parallel edges
of quadrilateral \quaa. Then, vertex $v_0$, which lies in the
interior of quadrilateral \quaa~and is incident to all vertices of
quadrilateral \quaa~cannot reside to none of \quaaa~and \quaab,
without introducing a non-right angle crossing with edge $e_q$.

We proceed to consider the case where quadrilaterals \quaa~and
\quab~cross and some of the vertices of quadrilateral \quab~are in
the interior of  quadrilateral \quaa, whereas the remaining ones on
its exterior. Let $q_a$, $q_b$, $q_c$ and $q_d$ be the vertices of
quadrilateral \quab. Assume that $q_a$ and $q_b$ lie within
quadrilateral \quaa, whereas $q_c$ and $q_d$ on its external face,
such that edges $(q_a,q_d)$ and $(q_b,q_c)$ are perpendicular either
to one edge of quadrilateral \quaa~(see
Figure~\ref{fig:v0internalq4external-2}), or to two edges of
\quaa~(see Figure~\ref{fig:v0internalq4external-3}). Note that edges
$(q_a,q_d)$ and $(q_b,q_c)$ cannot be crossed by any other edge
incident to both quadrilaterals, since they are already involved in
right-angle crossings. However, all vertices of quadrilateral
\quaa~have to be connected to vertex $v_0$. Assuming that one vertex
of quadrilateral \quaa~can utilize the ``last'' available edge of
quadrilateral \quab~(i.e., edge $(q_a,q_b)$) to reach vertex $v_0$
(refer to the dotted edges of
Figures~\ref{fig:v0internalq4external-2}~and~\ref{fig:v0internalq4external-3}),
there exists at least one vertex of \quaa, say vertex $q$, that
cannot be connected to $v_0$, without introducing non right-angle
crossing (refer to the dashed edges of
Figures~\ref{fig:v0internalq4external-2}~and~\ref{fig:v0internalq4external-3}).

\begin{figure}[h!tb]
  \centering
  \begin{minipage}[b]{.48\textwidth}
     \centering
     \subfloat[\label{fig:v0internalq4external-2}{}]
     {\includegraphics[width=.9\textwidth]{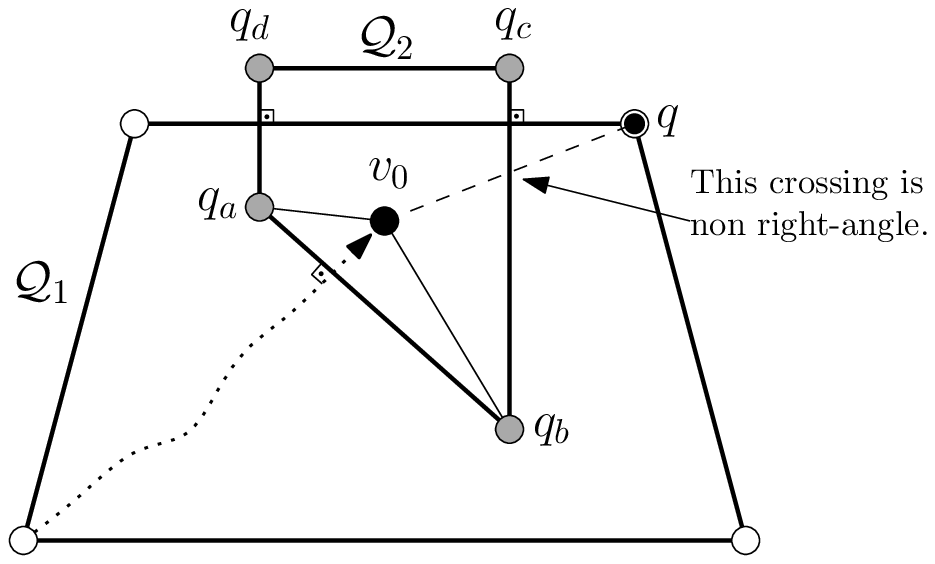}}
   \end{minipage}
  \hfill
  \begin{minipage}[b]{.48\textwidth}
     \centering
     \subfloat[\label{fig:v0internalq4external-3}{}]
     {\includegraphics[width=.7\textwidth]{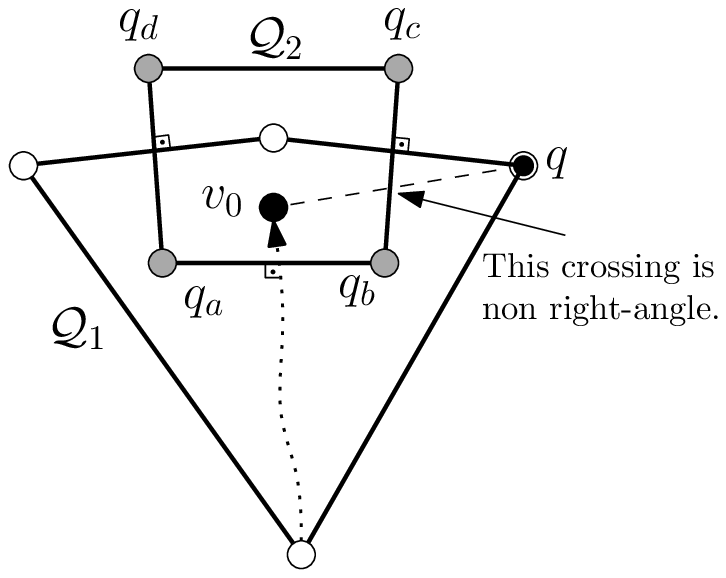}}
  \end{minipage}
  \caption{Vertices $q_a$ and $q_b$ are in the interior of \quaa~and \quab~crosses
  (i)~a single edge of \quaa, or (ii)~two edges of \quaa.}
  \label{fig:v0internalq4external-a}
\end{figure}

Following a similar reasoning scheme as for the previous cases, we
can prove that the cases where (i)~three vertices of \quab, say
w.l.o.g., $q_a$, $q_b$ and $q_c$, lie in the interior of \quaa~(see
Figure~\ref{fig:v0internalq4external-8}) , and (ii)~only one vertex
of \quab, say w.l.o.g, vertex $q_b$, lies in the interior of
\quaa~(see Figure~\ref{fig:v0internalq4external-9}), are led to a
contradiction. \qed

\begin{figure}[h!tb]
  \centering
  \begin{minipage}[b]{.48\textwidth}
     \centering
     \subfloat[\label{fig:v0internalq4external-8}{}]
     {\includegraphics[width=.7\textwidth]{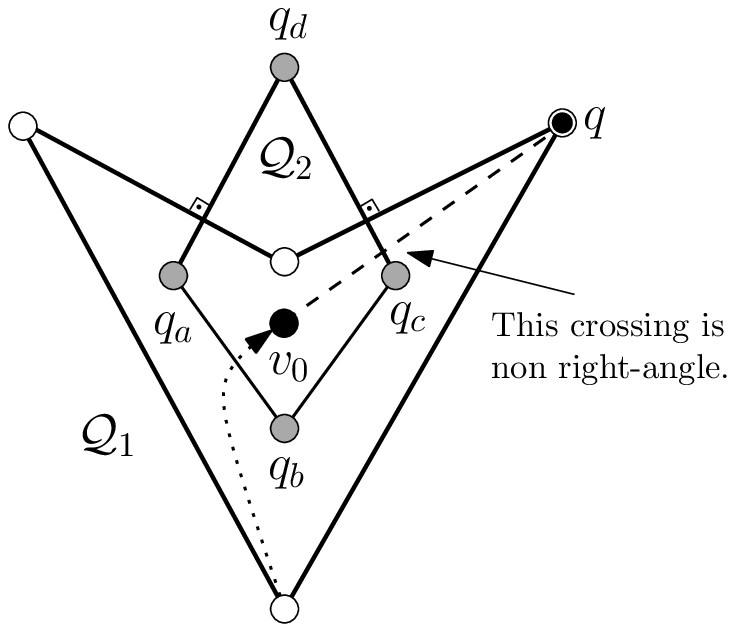}}
   \end{minipage}
  \hfill
  \begin{minipage}[b]{.48\textwidth}
     \centering
     \subfloat[\label{fig:v0internalq4external-9}{}]
     {\includegraphics[width=.6\textwidth]{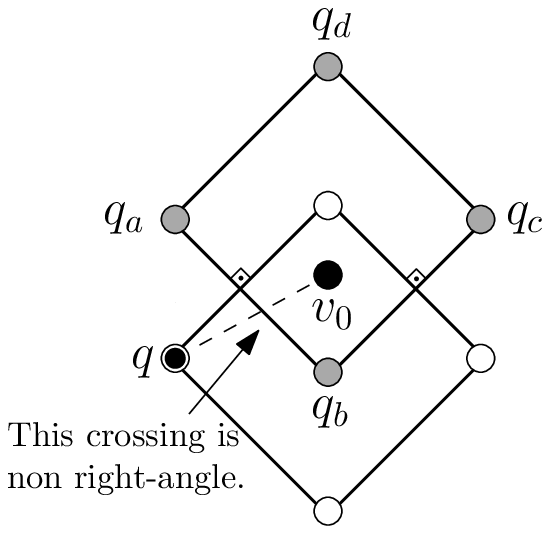}}
  \end{minipage}
  \caption{(i)~Vertices $q_a$, $q_b$ and $q_c$ are in the interior of \quaa. (ii)~Vertex $q_b$ is in the interior of \quaa.}
  \label{fig:v0internalq4external-a}
\end{figure}

\end{proof}


\begin{theorem}
\label{thm:uniqueness} Any straight-line RAC drawing of the
augmented square antiprism graph has two combinatorial embeddings.
\end{theorem}

\begin{proof}
So far, we have managed to prove that both quadrilaterals \quaa~and
\quab~are drawn planar, do not cross, and have central vertex $v_0$
to their interior. This suggests that either quadrilateral \quaa~is
in the interior of \quab, or quadrilateral \quab~is in the interior
of \quaa. However, in both cases, vertex $v_0$, which has to be
connected to the four vertices of the ``external'' quadrilateral,
should inevitably perpendicularly cross the four edges of the
``internal'' quadrilateral, and this trivially implies only two
feasible combinatorial embeddings. \qed
\end{proof}

We extend the augmented square antiprism graph, by appropriately
``glueing'' multiple instances of it, the one next to the other,
either horizontally or vertically.
Figure~\ref{fig:basic-gadget-extension} demonstrates how a
horizontal extension of two instances, say $G$ and $G'$, is
realized, i.e., by identifying two ``external'' vertices, say $v$
and $v'$, of $G$ with two ``external'' vertices of $G'$ (refer to
the gray-colored vertices of
Figure~\ref{fig:basic-gadget-extension}), and by employing an
additional edge (refer to the dashed drawn edge of
Figure~\ref{fig:basic-gadget-extension}), which connects an
``internal'' vertex, say $u$, of $G$ with the corresponding
``internal'' vertex, say $u'$, of $G'$. Let $G \oplus G'$ be the
graph produced by a horizontal or vertical extension of $G$ and
$G'$. Since each of $G$ and $G'$ has two RAC combinatorial
embeddings each, one would expect that $G \oplus G'$ would have four
possible RAC combinatorial embeddings. We will show that this is not
true and, more precisely, that there only exists a single RAC
combinatorial embedding.

\begin{figure}[h!tb]
  \centering
  \begin{minipage}[b]{\textwidth}
     \centering
     \subfloat[\label{fig:basic-gadget-extension}{}]
     {\includegraphics[width=\textwidth]{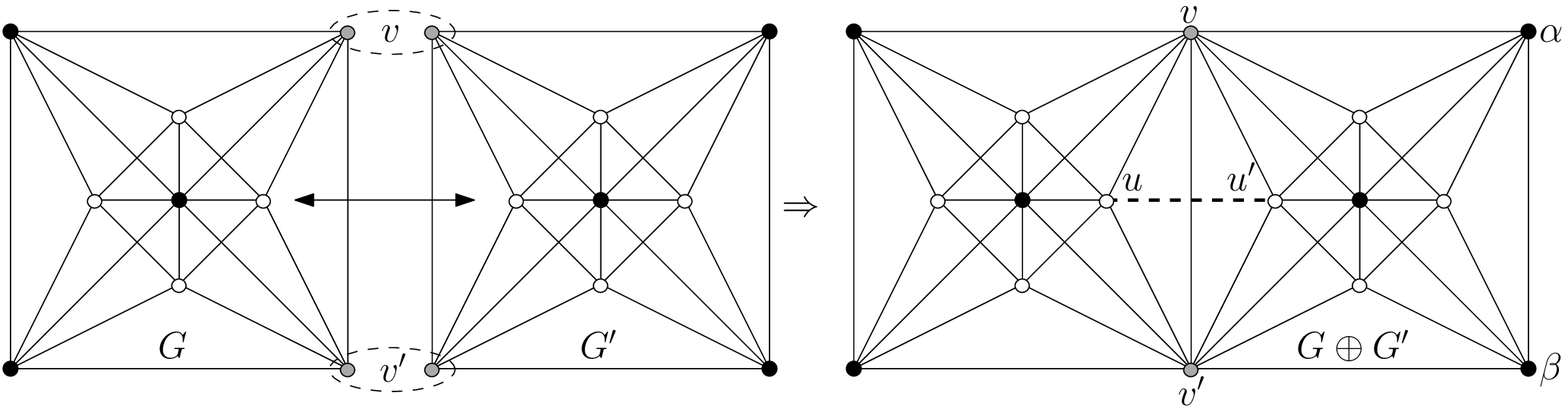}}
  \end{minipage}
  \vfill
  \begin{minipage}[b]{.48\textwidth}
     \centering
     \subfloat[\label{fig:basic-gadget-extension-no-cross}{}]
     {\includegraphics[width=.85\textwidth]{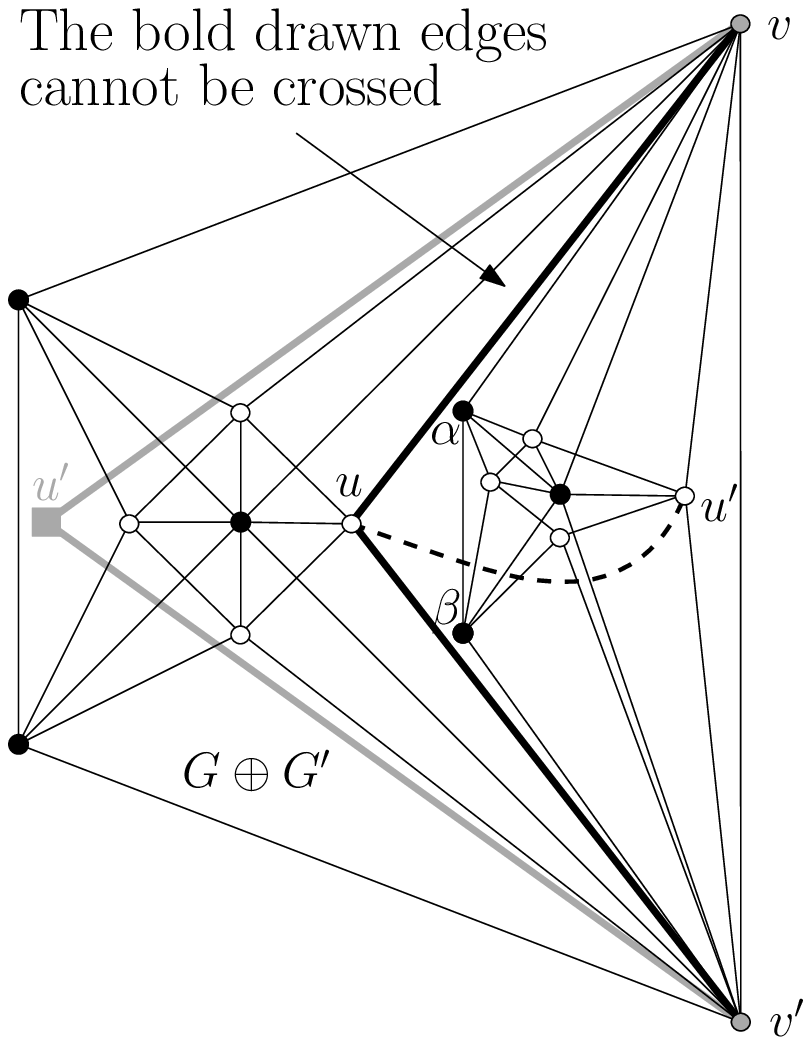}}
  \end{minipage}
  \hfill
  \begin{minipage}[b]{.48\textwidth}
     \begin{minipage}[b]{\textwidth}
        \centering
        \subfloat[\label{fig:basic-gadget-extension-internal}{}]
        {\includegraphics[width=.65\textwidth]{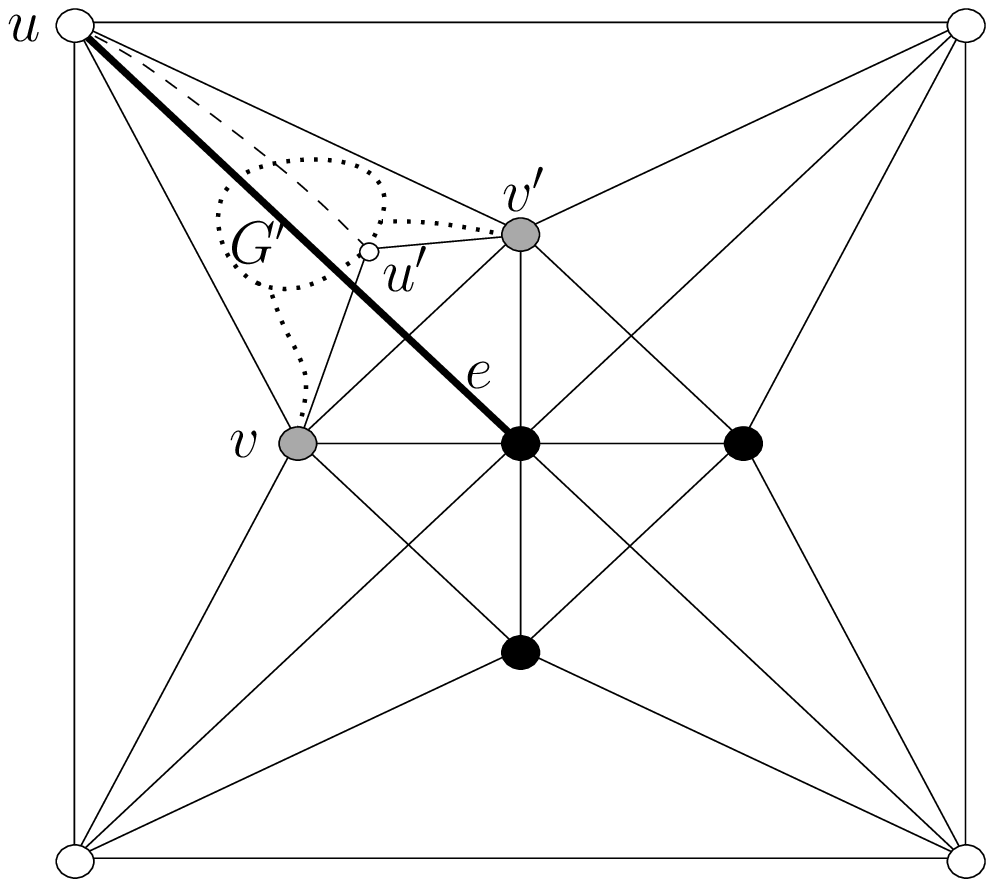}}
     \end{minipage}
     \begin{minipage}[b]{\textwidth}
        \centering
        \subfloat[\label{fig:basic-gadget-turned}{}]
        {\includegraphics[width=.85\textwidth]{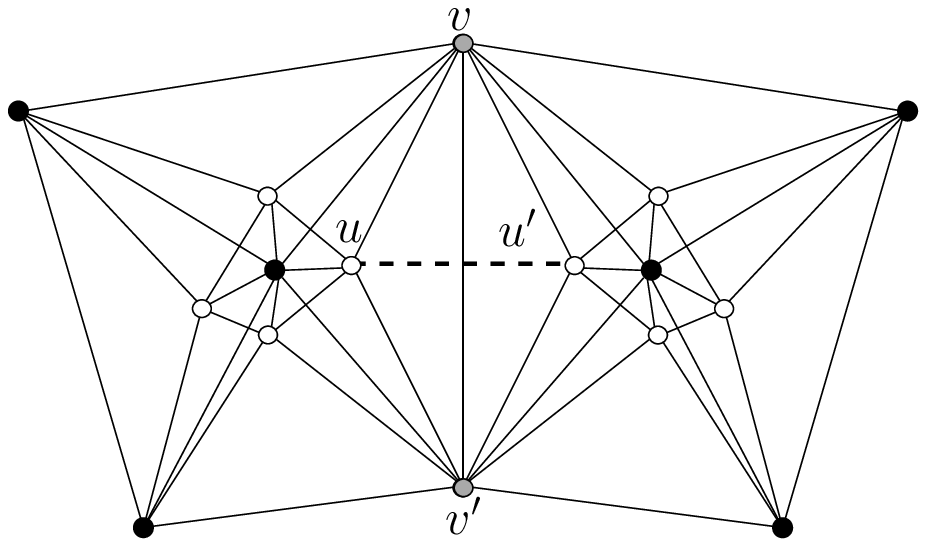}}
     \end{minipage}
  \end{minipage}
  \caption{(a)~Horizontal extension of two instances of the augmented square antiprism graph,
  (b)~The additional (dashed) edge does not permit the second instance to be drawn in the interior of the first one.
  (c)~The vertices which are identified, during a horizontal or vertical extension ($v$ and $v'$ in Figure), should be on the external face of each augmented square antiprism graph.
  (d)~At each extension step the new instance of the augmented square antiprism graph may introduce a ``turn''.}
  \label{fig:basic-gadget-attributes}
\end{figure}

\begin{theorem}
Let $G$ and $G'$ be two instances of the augmented square antiprism
graph. Then, $G \oplus G'$ has a unique RAC combinatorial embedding.
\end{theorem}

\begin{proof}
Assume first that in a RAC drawing of $G \oplus G'$, vertices $v$
and $v'$ are on the external quadrilateral of $G$ and graph $G'$ is
drawn completely in the interior of $G$ (see
Figure~\ref{fig:basic-gadget-extension-no-cross}; since $v$ and $v'$
are on the external face of $G'$, vertices $\alpha$ and $\beta$ in
Figure~\ref{fig:basic-gadget-extension-no-cross} should also be on
the external face of $G'$). First observe that vertex $u'$ of $G'$,
which is incident to vertices $v$ and $v'$, cannot reside to the
``left'' of both edges $(u,v)$ and $(u,v')$ (refer to the bold drawn
edges of Figure~\ref{fig:basic-gadget-extension-no-cross}), since
this would lead to a situation where three edges mutually cross and,
subsequently, to a violation of
Property~\ref{prp:three-crossing-edges} (see the gray-colored square
vertex of Figure~\ref{fig:basic-gadget-extension-no-cross}).
Therefore, vertex $u'$ should lie within the triangular face of $G$
formed by vertices $u$, $v$ and $v'$. The same similarly holds for
the central vertex of $G'$, which is also incident to vertices $v$
and $v'$. By Property~\ref{prp:triangle-edges}, any common neighbor
of vertices $u'$ and $v$ should also lie within the same triangular
face of $G$, which progressively implies that entire graph $G'$
should reside within this face, as in
Figure~\ref{fig:basic-gadget-extension-no-cross}. However, in this
case and since $u'$ is incident to $v$ and $v'$, edge $(u,u')$,
which is used on a horizontal or a vertical extension, crosses the
interior of $G'$, which is not permitted. This suggests that graph
$G'$ should be on the exterior of $G$.

Now assume that vertices $v$ and $v'$, which are identified, during
a horizontal or vertical extension, are along the internal
quadrilateral of $G$ in a RAC drawing of $G \oplus G'$. This is
illustrated in Figure~\ref{fig:basic-gadget-extension-internal}.
Then, the edge, say $e$, which perpendicularly crosses edge $(v,v')$
and emanates from the external quadrilateral towards the central
vertex of $G$ (refer to the bold solid edge of
Figure~\ref{fig:basic-gadget-extension-internal}) will be involved
in crossings with $G'$. More precisely, we focus on vertex $u'$ of
$G'$, which is incident to vertices $v$ and $v'$. These edges will
inevitably introduce non-right angle crossings, since one of them
should cross edge $e$. Therefore, the vertices that are identified,
during a horizontal or vertical extension, should always be on the
external face of each augmented square antiprism graph and,
subsequently, the drawing of the graph produced by a horizontal or
vertical extension will resemble the one of
Figure~\ref{fig:basic-gadget-extension}, i.e., it has a unique
embedding. \qed
\end{proof}

Note that the extension which is given in
Figure~\ref{fig:basic-gadget-extension}, is ideal. In the general
case, at each extension step the new instance of the augmented
square antiprism graph may introduce a ``turn'', as in
Figure~\ref{fig:basic-gadget-turned}. We observe that by ``glueing'' a new
instance of the augmented square antiprism graph on $G \oplus G'$ either by
a horizontal or a vertical extension, we obtain another graph of unique RAC combinatorial embedding.
In this way, we can define an infinite class of graphs of unique RAC combinatorial embedding. This is
summarized in the following theorem.

\begin{theorem}
There exists a class of graphs of unique RAC combinatorial embedding.
\end{theorem}

\section{The Straight-Line RAC Drawing Problem is NP-hard}
\label{sec:np}

\begin{theorem}
It is $\NP$-hard to decide whether an input graph admits a
straight-line RAC drawing.
\end{theorem}

\begin{proof}

We will reduce the well-known $3$-SAT problem \cite{GJ79} to the
straight-line RAC drawing problem. In a $3$-SAT instance, we are
given a formula $\phi$ in conjunctive normal form with variables
$x_1, x_2, \ldots, x_n$ and clauses $C_1, C_2, \ldots, C_m$, each
with three literals. We show how to construct a graph $G_\phi$ that
admits a straight-line RAC drawing $\Gamma(G_\phi)$ if and only if
formula $\phi$ is satisfiable.

Figure~\ref{fig:gadgets} illustrates the gadgets of our
construction. Each gray-colored square in these drawings corresponds
to an augmented square antiprism graph. Adjacent gray squares form
an extension (refer, for example, to the topmost gray squares of
Figure~\ref{fig:gadgets}a, which form a ``horizontal'' extension).
There also exist gray squares that are not adjacent, but connected
with edges. The legend in Figure~\ref{fig:gadgets} describes how the
connections are realized.

\begin{figure}[p]
   \centering
   \includegraphics[width=\textwidth]{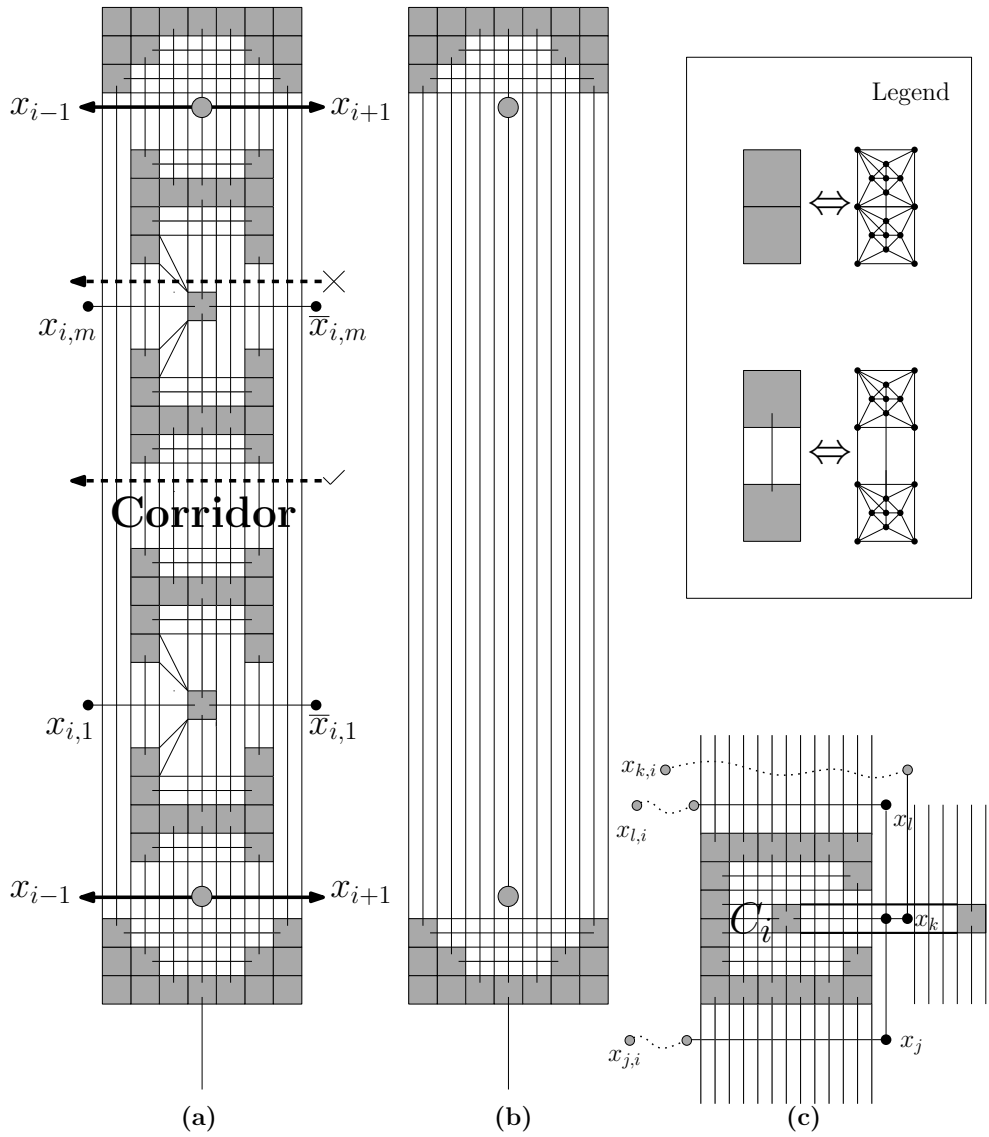}
   \caption{Gadgets of our construction: (a)~Variable gadget, (b)~Dummy variable gadget, (c)~Clause gadget}
   \label{fig:gadgets}
\end{figure}

The gadget that encodes variable $x_i$ of formula $\phi$ is given in
Figure~\ref{fig:gadgets}a. The gadget of variable $x_i$ consists of
a combination of augmented square antiprism graphs, and,
``horizontal'' and ``vertical'' edges, which form a tower, whose RAC
drawing has unique combinatorial embedding. One side of the tower
accommodates multiple vertices that correspond to literal $x_i$,
whereas its opposite side accommodates vertices that correspond to
literal $\overline{x_i}$ (refer to vertices $x_{i,1}, \ldots
,x_{i,m}$ and $\overline{x}_{i,1}, \ldots ,\overline{x}_{i,m}$ in
Figure~\ref{fig:gadgets}a). These vertices are called \emph{variable
endpoints}. Then, based on whether on the final drawing the negated
vertices will appear to the ``left'' or to the ``right'' side of the
tower, we will assign a true or a false value to variable $x_i$,
respectively. Pairs of consecutive endpoints $x_{i,j}$ and
$x_{i,j+1}$ are separated by a \emph{corridor} (see
Figure~\ref{fig:gadgets}a), which allows perpendicular edges to pass
through it (see the bottommost dashed arrow of
Figure~\ref{fig:gadgets}a). Note that this is not possible through a
``corridor'' formed on a variable endpoint, since there exist four
non-parallel edges that ``block'' any other edge passing through
them (see the topmost dashed arrow of Figure~\ref{fig:gadgets}a).
The corridors can have variable height. In the variable gadget of
variable $x_i$, there are also two vertices (they are drawn as gray
circles in Figure~\ref{fig:gadgets}a), which have degree four. These
vertices serve as ``\emph{connectors}'' among consecutive variable
gadgets, i.e., these vertices should be connected to their
corresponding vertices on the variable gadgets of variables
$x_{i-1}$ and $x_{i+1}$. Note that the connector vertices of the
variable gadgets associated with variables $x_1$ and $x_n$ are
connected to connectors of the variable gadgets that correspond to
variables $x_2$ and $x_{n-1}$, respectively, and to connectors of
\emph{dummy variable gadgets}.

Figure~\ref{fig:gadgets}b illustrates a dummy variable gadget, which
(similarly to the variable gadget) consists of a combination of
augmented square antiprism graphs, and, ``horizontal'' and
``vertical'' edges, which form a tower. Any RAC drawing of this
gadget has also unique combinatorial embedding. A dummy variable
gadget does not support vertices that correspond to literals.
However, it contains connector vertices (they are drawn as gray
circles in Figure~\ref{fig:gadgets}b). In our construction, we use
exactly two dummy variable gadgets. The connector vertices of each
dummy variable gadget should be connected to their corresponding
connector vertices on the variable gadgets associated with variables
$x_1$ and $x_n$, respectively.

The gadget that encodes the clauses of formula $\phi$ is illustrated
in Figure~\ref{fig:gadgets}c and resembles to a valve. Let $C_i=(x_j
\vee x_k \vee x_l)$ be a clause of $\phi$. As illustrated in
Figure~\ref{fig:gadgets}c, the gadget which corresponds to clause
$C_i$ contains three vertices\footnote{With slight abuse of
notation, the same term is used to denote variables of $\phi$ and
vertices of $G_\phi$.}, say $x_j$, $x_k$, and $x_l$, such that:
$x_j$ has to be connected to $x_{j,i}$, $x_k$ to $x_{k,i}$ and $x_l$
to $x_{l,i}$ \underline{by paths of length two}. These vertices,
referred to as the \emph{clause endpoints}, encode the literals of
each clause. Obviously, if a clause contains a negated literal, it
should be connected to the negated endpoint of the corresponding
variable gadget. The clause endpoints are incident to a vertex
``trapped'' within two parallel edges (refer to the bold drawn edges
of Figure~\ref{fig:gadgets}c). Therefore, in a RAC drawing of
$G_{\phi}$, only two of them can perpendicularly cross these edges,
one from top (\emph{top endpoint}) and one from bottom (\emph{bottom
endpoint}). The other one (\emph{right endpoint}) should remain in
the interior of the two parallel edges. The one that will remain
``trapped'' on the final drawing will correspond to the true literal
of this clause.


The gadgets, which correspond to variables and clauses of $\phi$,
are connected together by the skeleton of graph $G_{\phi}$, which is
depicted in Figure~\ref{fig:skeleton}a.  The skeleton consists of
two main parts, i.e., one ``horizontal'' and one ``vertical''. The
vertical part accommodates the clause gadgets (see
Figure~\ref{fig:skeleton}a). The horizontal part will be used in
order to ``plug'' the variable gadgets. The long edges that
perpendicularly cross (refer to the crossing edges slightly above
the horizontal part in Figure~\ref{fig:skeleton}a), imply that the
vertical part should be perpendicular to the horizontal part. The
horizontal part of the skeleton is separately illustrated in
Figure~\ref{fig:skeleton}b. Observe that it contains one set of
horizontal lines. 

\begin{figure}[htb]
   \centering
   \includegraphics[width=\textwidth]{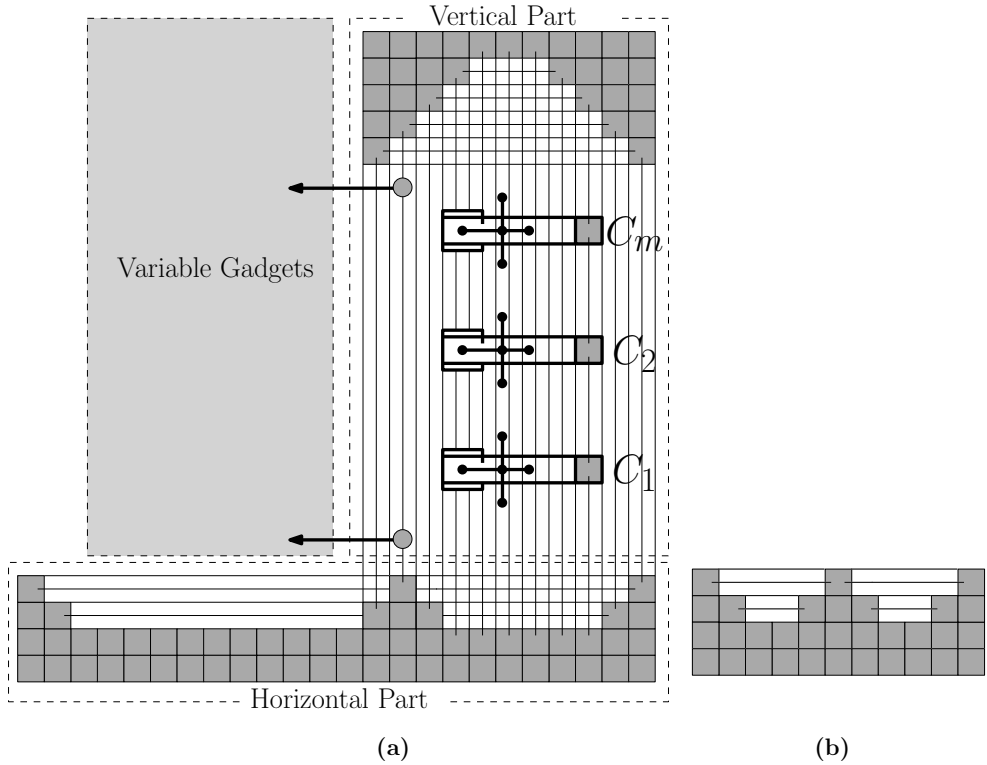}
   \caption{Illustration of the skeleton of the construction.}
   \label{fig:skeleton}
\end{figure}

Figure~\ref{fig:example} shows how the variable gadgets are attached
to the skeleton. More precisely, this is accomplished by a single
edge, which should perpendicularly cross the set of the horizontal
edges of the horizontal part. Therefore, each variable gadget is
perpendicularly attached to the skeleton, as in
Figure~\ref{fig:example}. Note that each variable gadget should be
drawn completely above of these horizontal edges, since otherwise
the connections among variable endpoints and clause endpoints would
not be feasible. The connector vertices of the dummy variable
gadgets, the variable gadgets and the vertical part of the
construction, ensure that the variable gadgets will be parallel to
each other (i.e., they are not allowed to bend) and parallel to the
vertical part of the construction.

\begin{figure}
   \centering
   \includegraphics[width=\textwidth]{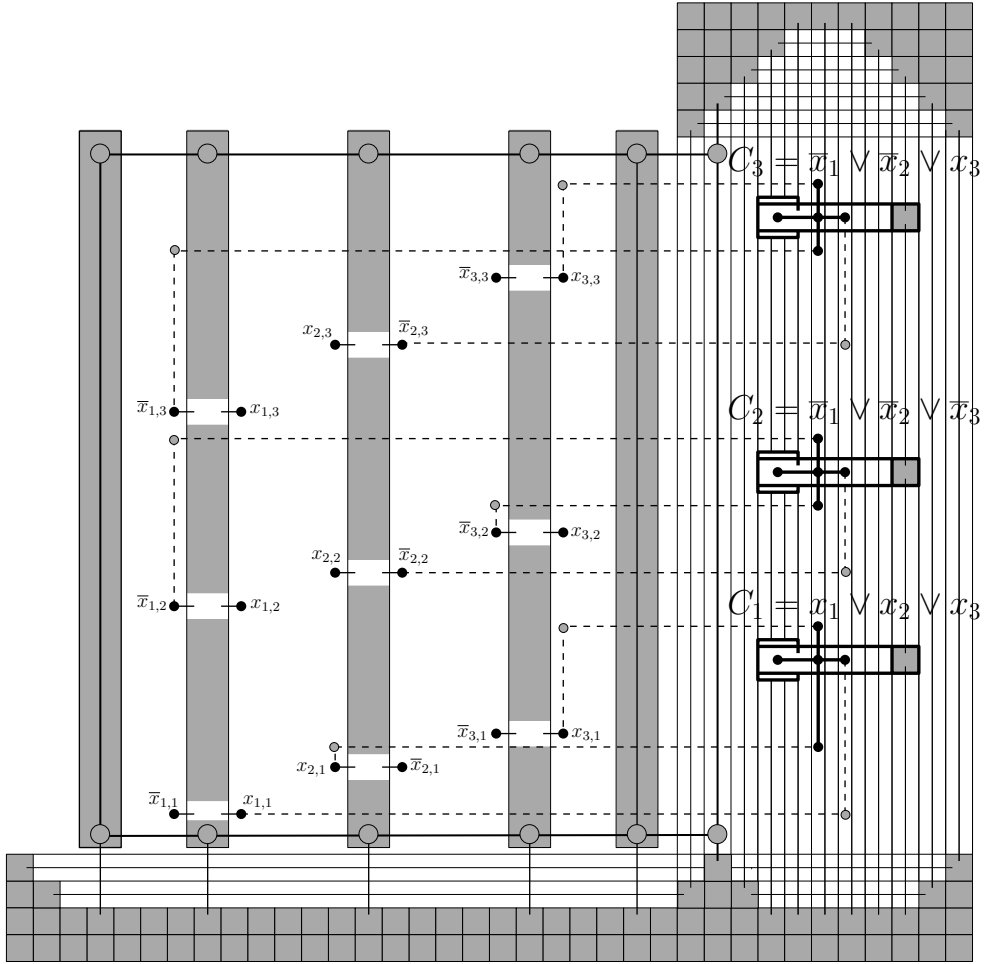}
   \caption{The reduction from $3$-SAT to the straight-line RAC drawing problem.
   The input formula is $\phi=(x_1 \vee x_2 \vee x_3)\wedge(\overline{x}_1 \vee
    \overline{x}_2 \vee \overline{x}_3)\wedge(\overline{x}_1 \vee
    \overline{x}_2 \vee x_3)$. The drawing corresponds to the truth
    assignment $x_1$=$x_3$=true, $x_2$=false.}
   \label{fig:example}
\end{figure}

We now proceed to investigate some properties of our construction.
Any path of length two that emanates from a top- or bottom-clause
endpoint can reach a variable endpoint either on the left or on the
right side of its associated variable gadget. The first edge of this
path should perpendicularly cross the vertical edges of the vertical
part of the construction and pass through some corridors\footnote{In
Figure~\ref{fig:example}, the corridors are the gray-colored regions
that reside at each variable gadget.}, whereas the second edge will
be used to realize the ``final'' connection with the variable gadget
endpoint (see Figure~\ref{fig:example}). However, the same doesn't
hold for the paths that emanate from a right-clause endpoint. These
paths can only reach variable endpoints on the right side of their
associated variable gadgets. More precisely, the first edge of the
$2$-length path should cross one of the two parallel edges (refer to
the bold drawn edges of Figure~\ref{fig:gadgets}c) that ``trap'' it,
whereas the other one should be used to reach (passing through
variable corridors) its variable endpoint (see
Figure~\ref{fig:example}).

Our construction ensures that up to translations, rotations and
stretchings any RAC drawing of $G_\phi$ resembles the one of
Figure~\ref{fig:skeleton}. It is clear that the construction can be
 completed in $O(nm)$ time. Assume now that there is a RAC drawing
$\Gamma(G_\phi)$ of $G_\phi$. If the negated vertices of the
variable gadget that corresponds to $x_i$, $i=1,2,\ldots,n$, lie to
the ``left'' side in $\Gamma(G_\phi)$, then variable $x_i$ is set to
true, otherwise $x_i$ is set to false. We argue that this assignment
satisfies $\phi$. To realize this, observe that there exist three
paths that emanate from each clause gadget. The one that emanates
from the right endpoint of each clause gadget can never reach a
false value. Therefore, each clause of $\phi$ must contain at least
one true literal, which implies that $\phi$ is satisfiable.

Conversely, suppose that there is a truth assignment that satisfies
$\phi$. We proceed to construct a RAC drawing $\Gamma(G_\phi)$ of
$G_\phi$, as follows: In the case where, in the truth assignment,
variable $x_i$, $i=1,2,\ldots,n$ is set to true, we place the
negated vertices of the variable gadget that corresponds to $x_i$,
to its left side in $\Gamma(G_\phi)$, otherwise to its right side.
Since each clause of $\phi$ contains at least one true literal, we
choose this as the right endpoint of its corresponding clause
gadget. As mentioned above, it is always feasible to be connected to
its variable gadgets by paths of length two. This completes our
proof. \qed
\end{proof}

\section{Conclusions}
\label{sec:conclusions}

In this paper, we proved that it is $\NP$-hard to decide whether a
graph admits a straight-line RAC drawing. Didimo et al.\
\cite{DEL09} proved that it is always feasible to construct a RAC
drawing of a given graph with at most three bends per edge. If we
permit two bends per edge, does the problem remain $\NP$-hard? It is
also interesting to continue the study on the interplay between the
number of edges and the required area in order to fill the gaps
between the known upper and lower bounds.

\bibliographystyle{abbrv}

\begin{thebibliography}{10}
\bibitem{ACBDFKS09}
Angelini, P., Cittadini, L., {Di Battista}, G., Didimo, W., Frati,
F.,
  Kaufmann, M., Symvonis, A.: On the perspectives opened by right angle
  crossing drawings. In: Proc. of 17th International Sympsioum on Graph Drawing
  (GD09). LNCS, vol. 5849, pp. 21--32 (2009)

\bibitem{ABS10}
Argyriou, E.N., Bekos, M.A., Symvonis, A.: Maximizing the total
resolution of
  graphs. In: Proc. of 18th International Sympsioum on Graph Drawing (GD10)
  (2010), to appear

\bibitem{AFKMT10}
Arikushi, K., Fulek, R., Keszegh, B., Moric, F., Toth, C.: Drawing
graphs with
  orthogonal crossings. In: Proc. 36th International Workshop on Graph
  Theoretic Concepts in Computer Science (WG 2010) (2010), to appear

\bibitem{DBTT94}
Battista, G.D., Eades, P., Tamassia, R., Tollis, I.G.: Algorithms
for drawing
  graphs: an annotated bibliography. Compututational Geometry  4,  235--282
  (1994)

\bibitem{BT04}
Bodlaender, H.L., Tel, G.: A note on rectilinearity and angular
resolution.
  Journal of Graph Algorithms and Applications  8,  89--94 (2004)

\bibitem{DGDLM10}
{Di Giacomo}, E., Didimo, W., Liotta, G., Meijer, H.: Area, curve
complexity,
  and crossing resolution of non-planar graph drawings. In: Proc. of 17th
  International Sympsioum on Graph Drawing (GD09). LNCS, vol. 5849, pp. 15--20
  (2009)

\bibitem{DEL09}
Didimo, W., Eades, P., Liotta, G.: Drawing graphs with right angle
crossings.
  In: Proc. of 12th International Symposium, Algorithms and Data Structures
  (WADS09). LNCS, vol. 5664, pp. 206--217 (2009)

\bibitem{DEL10}
Didimo, W., Eades, P., Liotta, G.: A characterization of complete
bipartite
  graphs. Information Processing Letters  110(16),  687--691 (2010)

\bibitem{DGMW10}
Dujmovic, V., Gudmundsson, J., Morin, P., Wolle, T.: Notes on large
angle
  crossing graphs. In: Computing: Theory of Computing 2010. Australian Computer
  Society (2010)

\bibitem{FHHKLSWW93}
Formann, M., Hagerup, T., Haralambides, J., Kaufmann, M., Leighton,
F.,
  Symvonis, A., Welzl, E., Woeginger, G.: Drawing graphs in the plane with high
  resolution. SIAM Journal of Computing  22(5),  1035--1052 (1993)

\bibitem{GJ79}
Garey, M.R., Johnson, D.S.: Computers and Intractability: A Guide to
the Theory
  of NP-Completeness. W. H. Freeman (1979)

\bibitem{GJ83}
Garey, M., Johnson, D.: Crossing number is {N}{P}-complete. SIAM
Journal of
  Algebraic Discrete Methods  4,  312--316 (1983)

\bibitem{GT94}
Garg, A., Tamassia, R.: Planar drawings and angular resolution:
Algorithms and
  bounds (extended abstract). In: Proc. 2nd Annual European Symposium on
  Algorithms. pp. 12--23 (1994)

\bibitem{GM98}
Gutwenger, C., Mutzel, P.: Planar polyline drawings with good
angular
  resolution. In: Proc. of 6th International Symposium on Graph Drawing. LNCS,
  vol. 1547, pp. 167--182 (1998)

\bibitem{Hu07}
Huang, W.: Using eye tracking to investigate graph layout effects.
Asia-Pacific
  Symposium on Visualization pp. 97--100 (2007)

\bibitem{HHE08}
Huang, W., Hong, S.H., Eades, P.: Effects of crossing angles. In:
PacificVis.
  pp. 41--46 (2008)

\bibitem{KW01}
Kaufmann, M., Wagner, D. (eds.): Drawing Graphs: Methods and Models,
LNCS, vol.
  2025. Springer-Verlag (2001)

\bibitem{vK10}
van Kreveld, M.: The quality ratio of {RAC} drawings and planar
drawings of
  planar graphs. In: Proc. of 18th International Sympsioum on Graph Drawing
  (GD10) (2010), to appear

\bibitem{LY05}
Lin, C.C., Yen, H.C.: A new force-directed graph drawing method
based on
  edge-edge repulsion. In: Proc. of the 9th International Conference on
  Information Visualization. pp. 329--334. IEEE (2005)

\bibitem{MP92}
Malitz, S.M., Papakostas, A.: On the angular resolution of planar
graphs. In:
  Proc. of the 24th Annual ACM Symposium on Theory of Computing (STOC92). pp.
  527--538. ACM (1992)

\bibitem{P00}
Purchase, H.C.: Effective information visualisation: a study of
graph drawing
  aesthetics and algorithms. Interacting with Computers  13(2),  147--162
  (2000)

\bibitem{PCA02}
Purchase, H.C., Carrington, D.A., Allder, J.A.: Empirical evaluation
of
  aesthetics-based graph layout. Empirical Software Engineering  7(3),
  233--255 (2002)

\bibitem{CPCM02}
Ware, C., Purchase, H., Colpoys, L., McGill, M.: Cognitive
measurements of
  graph aesthetics. Information Visualization  1(2),  103--110 (2002)
\end{thebibliography}

\end{document}